\documentclass{birkjour}

\usepackage{times,amssymb,amsmath,epsfig,dcolumn,multirow,bm}
\usepackage[matrix,arrow]{xy}
\usepackage{placeins}

\numberwithin{equation}{section} 

\newcommand{\prb}{Phys.\ Rev.\ B }

\newcommand{\prl}{Phys.\ Rev.\ Lett.\ }

\newtheorem{theorem}{Theorem}[section]
\newtheorem{proposition}[theorem]{Proposition}
\newtheorem{definition}[theorem]{Definition}
\newtheorem{lemma}[theorem]{Lemma}

\newcommand{\trace}{\mathop{\rm trace}}
\newcommand{\mult}{\mathop{\rm mult}}
\newcommand{\vol}{\mathop{\rm vol}}

\newcommand{\longtwoheadrightarrow}{\relbar\joinrel\twoheadrightarrow}

\newcommand{\Sp}{\mathop{\boldsymbol{S}\hskip -0.05ex\boldsymbol{p}}}
\newcommand{\SU}{\mathop{\mbox{\it\bfseries SU}}}
\newcommand{\Orth}{\boldsymbol{O}}
\newcommand{\SOrth}{\mathop{\mbox{\it\bfseries SO}}}
\newcommand{\Simple}{\mathop{\mbox{\it\bfseries S}}}
\newcommand{\GL}{\mathop{\mbox{\it\bfseries GL}}}
\newcommand{\symp}{\mathop{\mbox{\boldmath $\mathfrak{sp}$}}}
\newcommand{\quat}{\boldsymbol{Q}_8}

\newcommand{\Herm}{\mathop{\rm Herm}}

\begin{document}

\title{\raggedright Pseudorotations in Molecules: \mbox{Electronic Orbital Triplets}}

\author{A.~R.~Rutherford}

\address{Department of Mathematics\\
         Simon Fraser University\\
         8888 University Drive\\
         Burnaby, B.~C.~ V5A 1S6 Canada}

\email{arruther@sfu.ca}

\date{October 22, 2017}

\dedicatory{\raggedright%
This work had its genesis in a collaboration with
Roy~R.~Douglas on applications of algebraic topology to the
Jahn-Teller effect and related phenomena in multi-body quantum
mechanics.  This paper is dedicated to his memory.}

\maketitle

\begin{abstract}
  Topological and geometrical methods are used to calculate the
  pseudorotational part of the vibronic spectrum for an electronic
  triplet of an octahedral, tetrahedral, or icosahedral molecule.  The
  calculations take into account the nontrivial geometry inherent in
  the Jahn-Teller effect.  It is shown that the Jahn-Teller effect
  gives rise to a geometry, which is related to the isoparametric
  geometry of E.~Cartan.  The pseudorotational spectra correspond to
  the spectra of connection Laplacians on nontrivial line bundles over
  base spaces with this geometry. Globally, the isoparametric
  submanifolds form a totally geodesic foliation of $S^4$ and the
  spectral flow of these connection Laplacians on this foliation is
  computed.
\end{abstract}

\tableofcontents

\section{Introduction}

The Jahn-Teller effect for electronic orbital triplets can occur for
molecules with tetrahedral, octahedral or icosahedral symmetry.  The
main consequence of the Jahn-Teller effect is that the minimal energy
configuration for such a molecule in an orbital triplet state will not
be when the nuclei are in the completely symmetric configuration.
Instead, the minimal energy occurs on an extended compact manifold of
asymmetric configurations.  Pseudorotations are vibronic excitations
arising from free distortions of the molecular geometry on this
compact manifold.  The most general pseudorotational modes of
electronic triplets involve a 5\nolinebreak\mbox{-}dimensional space
of normal mode coordinates for the positions of the nuclei within the
molecule.  These pseudorotations appear in the $T \otimes
(e \oplus t)$ Jahn-Teller effect of tetrahedral molecules, the $T
\otimes (e_g \oplus t_{2g})$ of octahedral molecules and the $T
\otimes h_g$ Jahn-Teller effect of icosahedral molecules.

An important example of the Jahn-Teller effect for tetrahedral
symmetry occurs for the methane cation ${\rm C}{\rm H}_4^+$.  The
lowest electronic orbital state of ${\rm C}{\rm H}_4^+$ is known to be
a triplet and the resulting $T \otimes (e \oplus t_2)$ Jahn-Teller
effect leads to a distortion of the molecule away from tetrahedral
symmetry~\cite{RBWKS71,KSFD84,PFPHP85,VKBCFKZZ86,T86,FD88,RD95}.
Octahedral symmetry is common in nature and the $T \otimes (e_g \oplus
t_{2g})$ Jahn-Teller effect has been studied
extensively~\cite{O64,H65,O69,O71,CBV84,CO88}.  For example, it
occurs in cubic crystals such as ${\rm Ca}{\rm O}{:}{\rm
Fe}^{2+}$~\cite{H70,J79} and ${\rm Mg}{\rm O}{:}{\rm
Fe}^{2+}$~\cite{HNV75}.  Also, it pays an important role in transition
metal perovskite crystals such as ${\rm Ba}{\rm Ti}{\rm
O}_3$~\cite{BP89} and ${\rm K}{\rm Fe}^{2+}{\rm
F}_3$~\cite{KK82,RRLR73} and triangular lattice Heisenberg
antiferromagnetic crystals such as ${\rm Li}{\rm V}{\rm O}_2$ and
${\rm Na}{\rm Ti}{\rm O}_2$~\cite{PBKS96}.  The importance of the
Jahn-Teller effect in perovskites suggests that it may be important
for high $T_c$ superconductivity.  Another example of the $T \otimes
(e_g \oplus t_{2g})$ Jahn-Teller effect are transition metal
hexafluorides such as ${\rm Re}{\rm F}_6$ and ${\rm Os}{\rm
F}_6$~\cite{HSHLO88,MSI96}.  To date, the most famous example of an
icosahedral molecule is the buckminsterfullerene molecule ${\rm
C}_{60}$.  The $T \otimes h_g$ Jahn-Teller effect for this molecule is
discussed in~\cite{GAS92,AMT94,MTA94}.  This brief survey of triplet
Jahn-Teller effects is by no means exhaustive and we refer
to~\cite{He66,E72,B84a,B84,J84,BP89,KV95} for more complete reviews.

The role of topology in the Jahn-Teller effect was first demonstrated
in calculations of the $E \otimes e$ Jahn-Teller pseudorotations in
triangular molecules~\cite{LOPS58,TTM85,TM85,TILTM85,H87}.  Moreover,
these calculations have been experimentally confirmed
in~\cite{DGWWZ86}.  It was realized that $E \otimes e$
pseudorotational wave functions are sections in the M\"obius band, the
unique nontrivial real line bundle over the circle.  The simplicity of
the M\"obius band allowed the calculation to be treated as a boundary
value problem.  In other words, it sufficed to solve for wave
functions on the interval $[0,2\pi]$ satisfying $\psi(0) =
-\psi(2\pi)$, rather than the usual boundary condition $\psi(0) =
\psi(2\pi)$ for wave functions on a circle.  We will see that topology
permeates the Jahn-Teller problem for orbital triplets to a much
greater extent and boundary value methods no longer suffice.

In order to determine the topology of the space of nuclear
pseudorotations for an orbital triplet state, we first study spaces of
$3 \times 3$ real symmetric matrices which have eigenvalues of
predetermined multiplicities.  The space of pseudorotations is
isometric to a manifold constructed from the electronic eigenstates.
A detailed description of both the topology and the Riemannian
geometry of this manifold is presented.  As a consequence of the
nontrivial topology of the manifold of pseudorotations, the usual
Born-Oppenheimer approximation cannot be applied to this problem.
However, using a generalisation of the Born-Oppenheimer approximation
to nontrivial vector bundles, we are able to compute vibronic energy
levels in the strong Jahn-Teller coupling limit.  In Section~V, we use
our results to compute the pseudorotational spectrum of ${\rm C}{\rm
H}_4^+$.  Our calculations are compared to measurements of the
vibronic lines in the photoelectron spectrum of methane reported
in~\cite{RBWKS71}. 

\section{Born-Oppenheimer Approximation}

Consider a molecule with $a$ nuclei of mass $M$ and $z$ electrons of
mass $m$.  Neglecting contributions from electron spin, the molecular
Hamiltonian is
\begin{equation}
  H_{\rm mol}
      = -\frac{1}{2M}\sum_{j=1}^a
          \triangle_{\mathbf{a}_j}
      - \frac{1}{2m}\sum_{i=1}^z \triangle_{\mathbf{z}_i}
      + V(\mathbf{a}, \mathbf{z}) 
      + V_{\rm nuc}(\mathbf{a}) \, ,
  \label{mol-ham}
\end{equation}
where $\mathbf{a}_j$ is the position vector for the $j^{\rm th}$
nucleus, $\mathbf{z}_i$ is the position vector for the $i^{\rm th}$
electron, and $\triangle_\mathbf{x}$ denotes the Laplacian for a
coordinate vector $\mathbf{x}$.  Collectively, the coordinate vectors
$\mathbf{a}_j$ define a total configuration vector $\mathbf{a}
\in \mathbb{R}^{3a}$ for the nuclei and the $\mathbf{z}_i$ define a
total configuration vector $\mathbf{z} \in \mathbb{R}^{3z}$ for the
electrons.  The term $V(\mathbf{a}, \mathbf{z})$ is the potential
energy for all electron-electron and electron-nucleus interactions and
$V_{\rm nuc}(\mathbf{a})$ is the potential energy for all
nucleus-nucleus interactions.  These potentials $V$ and $V_{\rm nuc}$
are the summations over the usual sets of two-body Coulomb potential
functions.  For simplicity we have assumed that all of the nuclei have
the same mass; however, this assumption is not essential for the
results which we obtain.

Although the nuclear mass is at least four orders of magnitude larger
than the electron mass, it has long been known that generally
molecular vibrations cannot be described in terms of classical nuclear
motion.  Instead an understanding of molecular energy levels must take
into account the quantum mechanical coupling between the electrons and
the nuclei\footnote{The term vibronic (vibrational~$+$~electronic) is
usually used to describe this coupling.}.  To this end, the
Hamiltonian $H_{\rm mol}$ is usually treated using the
Born-Oppenheimer approximation~\cite{BO27,CDS81,GH87,GH88}.  In
this approximation, one considers multiplets of
\mbox{$\mathbf{a}$-dependent} eigenvalues of the electronic
Hamiltonian
\begin{equation}
  H_{\rm el}(\mathbf{a})
       = -\frac{1}{2m} \sum_{i=1}^z \triangle_{\mathbf{z}_i} 
         + V(\mathbf{a}, \mathbf{z}) \, .
\end{equation}
Taking the direct sum of all eigenspaces associated with the
eigenvalues in a multiplet defines for each multiplet a
\mbox{$\mathbf{a}$-dependent} vector subspace of
$L^2(\mathbb{R}^{3z};\,\mathbb{C}\,)$, the Hilbert space of
complex-valued square-integrable functions on $\mathbb{R}^{3z}$.  For
each multiplet, an operator called the Born-Oppenheimer Hamiltonian is
obtained by expanding $H_{\rm mol}$ with respect to a
$\mathbf{a}$\nolinebreak\mbox{-}dependent orthonormal basis for
this subspace.  At sufficiently low energy, the point spectrum of
$H_{\rm mol}$ is approximated by the union of the point spectra of the
Born-Oppenheimer Hamiltonians obtained by subdividing the point
spectrum of $H_{\rm el}(\mathbf{a})$ into disjoint multiplets.  A
complete description of how much of the point spectrum of $H_{\rm
mol}$ may be approximated in this manner is too detailed for inclusion
here and we refer the reader to~\cite{CDS81,GH87,GH88}.

It follows from established results~\cite{W59} that the eigenspaces of
the electronic Hamiltonian are classified by irreducible
representations of the molecular symmetry group. In this paper, we are
considering 3\nolinebreak\mbox{-}dimensional eigenspaces of $H_{\rm
el}$, which are possible for molecules having tetrahedral, octahedral
or icosahedral symmetry.  Denoting the fully symmetric nuclear
configuration vector by $\mathbf{a}_0 \in \mathbb{R}^{3a}$, the
eigenspaces of $H_{\rm el}(\mathbf{a}_0)$ are representation spaces
of the molecular symmetry group.  This group is either the tetrahedral
group $\mathcal{T}_d$, the extended octahedral group $\mathcal{O}_h$,
or the extended icosahedral group $\mathcal{Y}_h$.

The 24\nolinebreak\mbox{-}element group $\mathcal{T}_d$ is the complete
symmetry group of a regular tetrahedron.  In a common nomenclature,
the irreducible representations of $\mathcal{T}_d$ are $A_1$ and $A_2$ of
dimension one, $E$ of dimension two, and $T_1$ and $T_2$ of dimension
three.  Reference~\cite{LF96} may be consulted for a review of
the representation theory of molecular symmetry groups.

The complete symmetry group of a regular octahedron is
$\mathcal{O}_h$, which is isomorphic to the direct product of the
octahedral group $\mathcal{O}$ with the 2\nolinebreak\mbox{-}element
group $\mathbb{Z}_2$.  The octahedral group $\mathcal{O}$ consists of
the $24$ rotations about the symmetry axes of a regular octahedron,
whereas the extended octahedral group also includes reflections in the
plane perpendicular to each symmetry axis.  The irreducible
representations of $\mathcal{O}$ are $A_1$ and $A_2$ of dimension one,
$E$ of dimension two, and $T_1$ and $T_2$ of dimension three.  Each
irreducible representation of $\mathcal{O}$ determines two irreducible
representations of $\mathcal{O}_h$.  One of these, denoted by the
subscript $u$, is odd under inversion about the origin and the other,
denoted by the subscript $g$, is even under inversion.

The complete symmetry group of a regular icosahedron is the extended
icosahedral group $\mathcal{Y}_h$.  This is a
120\nolinebreak\mbox{-}element group, which is isomorphic to the
direct product of $\mathbb{Z}_2$ with the icosahedral group $\mathcal{Y}$.
The irreducible representations of $\mathcal{Y}$ are $A$ of dimension
one, $T_1$ and $T_2$ of dimension three, $G$ of dimension four, and
$H$ of dimension five.  As with $\mathcal{O}_h$ above, each irreducible
representation of $\mathcal{Y}$ determines two irreducible
representations of $\mathcal{Y}_h$, one of which is even under inversion
and the other of which is odd under inversion.

We are considering a triplet eigenvalue of $\lambda$ of $H_{\rm
el}(\mathbf{a}_0)$.  The associated eigenspace is classified by a
3\nolinebreak\mbox{-}dimensional irreducible representation $T$ of the
molecular symmetry group.  This representation must be one of the
$T_1$ or $T_2$ representations of $\mathcal{T}_d$, the $T_{1u}$,
$T_{1g}$, $T_{2u}$, or $T_{2g}$ representations of $\mathcal{O}_h$, or
the $T_{1u}$, $T_{1g}$, $T_{2u}$ or $T_{2g}$ representations of
$\mathcal{Y}_h$.  Of course, nuclear motion will not preserve the
molecular symmetry and the degenerate eigenvalue $\lambda$ will be
split into an almost degenerate triplet $\{\lambda_1
(\mathbf{a}) \,,\, \lambda_2(\mathbf{a})
\,,\, \lambda_3(\mathbf{a}) \}$.  

A convenient coordinate system for describing displacements of the
nuclei about the symmetric configuration $\mathbf{a}_0$ are normal
mode coordinates~\cite{LL77}.  The nuclear configuration is specified
by a normal mode coordinate vector $\mathbf{q} \in \mathbb{R}^{3a -
6}$, where we have followed the usual procedure of eliminating the the
coordinates for rigid rotations and translations of the molecule.  In
this coodinate system, the symmetric configuration of the nuclei
corresponds to $\mathbf{q} = \mathbf{0}$.  The advantage of normal
mode coordinates is that they are also classified by the irreducible
representations of the molecular symmetry group~\cite{W30,LL77}.
Reducing the total vibrational representation of the molecular
symmetry group into $l$ irreducible representations of dimensions
$d_k$, for $k = 1,2,\dots,l$, induces a decomposition of the vector
$\mathbf{q}$ as
\begin{equation}
  \mathbf{q} = \left(q_1^1,q_2^1,\dots,q_{d_1}^1,q_1^2,q_2^2,\dots,
            q_{d_2}^2,\dots,q_1^l,q_2^l,\dots,q_{d_l}^l\right).
\end{equation}
Each subvector $\mathbf{q}^k \in \mathbb{R}^{d_k}$ transforms according to
one of the irreducible representations in the total vibrational
representation.

It was shown by Jahn and Teller~\cite{JT37} that the symmetric
configuration is not a simultaneous minimum of the eigenvalues
$\lambda_1(\mathbf{q})$, $\lambda_2(\mathbf{q})$ and
$\lambda_3(\mathbf{q})$.  They showed that for all molecules except
linear molecules\footnote{These are molecules having the axial
symmetry groups $\mathcal{C}_{\infty v}$ and $\mathcal{D}_{\infty
h}$.}, 
\begin{equation}
  \frac{\partial \lambda_1}{\partial q_j^k}(\mathbf{0}) =
  \frac{\partial \lambda_2}{\partial q_j^k}(\mathbf{0}) = \frac{\partial
  \lambda_3}{\partial q_j^k}(\mathbf{0}) = 0,\quad\text{for all $j$ and $k$,}
  \label{jt}
\end{equation}
if and only if the product representation $V\left[T^2\right]$ does not
contain the identity representation.  Here $[T^2]$ denotes the
symmetric square of the electronic representation $T$ and $V$ denotes
the vibrational representation.  The complete list of these product
representations for all molecular symmetry groups was decomposed
in~\cite{JT37} and it was found that with the exception of the axial
symmetry groups, $V\left[T^2\right]$ always contains the identity
representation.\footnote{It is now known how to prove the Jahn-Teller
theorem without resorting to an itemisation of all irreducible
representations of all molecular symmetry groups~\cite{RS65,B71}.
These more sophisticated proofs shed light on the mathematical content
of the Jahn-Teller theorem.}

For the remainder of this paper we will use the labelling convention that
\begin{equation}
  \lambda_1(\mathbf{q}) \le 
  \lambda_2(\mathbf{q}) \le
  \lambda_3(\mathbf{q})\,. \label{order}
\end{equation}
Note that the imposition of this convention will mean that the
$\lambda_i$ will not in general be differentiable at $\mathbf{0}$,
because the eigenvalues will be re-ordered as $\mathbf{q}$ passes
through the point of degeneracy at $\mathbf{0}$.

For the symmetry group $\mathcal{T}_d$, Jahn and Teller concluded that
if the electronic irreducible representation $T$ is either of the
representations $T_1$ or $T_2$, then the molecular symmetry will not
be stable with respect to molecular distortions in normal modes
classified by either of the $E$ and $T_2$ irreducible representations.
In otherwords, equation~(\ref{jt}) will not hold for $\mathbf{q}^k$
corresponding to either of these representations.  This Jahn-Teller
effect is called the $T \otimes \left(e_g \oplus t_{2g}\right)$
Jahn-Teller effect, where we are following the convention in the
literature of using lower case to denote irreducible representations
when they classify normal modes and upper case to denote irreducible
representations when they classify eigenspaces of $H_{\rm
el}(\mathbf{a}_0)$.  The $e$ and $t_2$ normal modes are said to be
Jahn-Teller coupled to the $T$ electronic state.  Similarly, if $T$ is
one of the 3\nolinebreak\mbox{-}dimensional irreducible
representations $T_{1g}$, $T_{1u}$, $T_{2g}$ or $T_{2u}$ of the group
extended octahedral group $\mathcal{O}_h$, then normal mode
distortions classified by both the $E_g$ and $T_{2g}$ representations
were found to be Jahn-Teller coupled to $T$.  This is the $T \otimes
(e_g \oplus t_{2g})$ Jahn-Teller effect of the molecular symmetry
group $\mathcal{O}_h$.  For the extended icosahedral group
$\mathcal{Y}_h$, if $T$ is one of the 3\nolinebreak\mbox{-}dimensional
irreducible representations $T_{1g}$, $T_{1u}$, $T_{2g}$ or $T_{2u}$,
then normal modes classified by the $H_g$ representation are
Jahn-Teller coupled to $T$.  This is the $T \otimes h_g$ Jahn-Teller
effect.  Our results in this paper apply to all three of these
Jahn-Teller effects.

In studying low energy excitations of the molecule, we will restrict
the nuclear configuration space $\mathbb{R}^{3a - 6}$ to the
5\nolinebreak\mbox{-}dimensional subspace $N \subset \mathbb{R}^{3a -
6}$ in which only the Jahn-Teller coupled normal modes are nonzero.
Furthermore, we shall restrict $\mathbf{q}$ to a suitably bounded
neighbourhood $N^\prime$ of $\mathbf{0}$ in the vector space $N$.
Specifically, $N^\prime$ should be starlike\footnote{A subset $X
\subset \mathbb{R}^n$ is called starlike from the point $x \in X$ if
for every $y \in X$ the line segment from $x$ to $y$ lies entirely in
$X$~\cite{S66}.} from $\mathbf{q} = \mathbf{0}$, which implies that
$N^\prime$ is contractible.  An appropriate bound on the diameter of
$N^\prime$ ensures that the triplet $\{\lambda_1 (\mathbf{q}) \,,\,
\lambda_2 (\mathbf{q}) \,,\, \lambda_3(\mathbf{q}) \}$ will be bounded
away from the remainder of the spectrum of $H_{\rm el}(\mathbf{q})$,
for all $\mathbf{q} \in N^\prime$.  This amounts to assuming that we
are studying distortions in which the molecule remains reasonably
close to its symmetric configuration.  In other words, we are
considering molecular distortions, rather than some more general
many-body problem.

Define $Z(\mathbf{q})$ to be the 3\nolinebreak\mbox{-}dimensional
direct sum of the eigenspaces for the triplet $\{\lambda_1(\mathbf{q})
\,,\, \lambda_2 (\mathbf{q}) \,,\, \lambda_3(\mathbf{q}) \}$.  Note
that for some nonzero values of $\mathbf{q} \in N^\prime$, the
eigenvalues $\lambda_i(\mathbf{q})$ will have twofold degeneracies.
Contractibility of $N^\prime$ implies that there exists a globally
defined smooth basis $\{ e_1(\mathbf{q}) \,,\, e_2(\mathbf{q}) \,,\,
e_3(\mathbf{q}) \}$ for $Z(\mathbf{q})$ over all $\mathbf{q} \in
N^\prime$.  However, we shall see that it is not possible to choose
this basis such that even one of the basis vectors $e_i(\mathbf{q})$
is an eigenvector of $H_{\rm el}(\mathbf{q})$ for all $\mathbf{q} \in
N^\prime$.

The basis $\{ e_1(\mathbf{q}) \,,\, e_2(\mathbf{q}) \,,\,
e_3(\mathbf{q}) \}$ defines a continuous map $H_Z$ from $N^\prime$
to\linebreak[4] $\mathop{\rm Herm}(3,\mathbb{C}\,)$, the
9\nolinebreak\mbox{-}dimensional real vector space of $3 \times 3$
hermitian matrices with complex entries.  The $(i,j)$ entry of the
matrix $H_Z(\mathbf{q})$ is defined by the inner product
\begin{equation}
  H_Z(\mathbf{q})_{ij} 
    = (e_i(\mathbf{q})\,,\,H_{\rm el}(\mathbf{q})\, 
      e_j(\mathbf{q}))\,,
  \label{Hz}
\end{equation}
where $(\,\cdot\,,\,\cdot\,)$ is the Hilbert space inner
product on $L^2(\mathbb{R}^{3z};\,\mathbb{C})$.  If the Hamiltonian
$H_{\rm el}(\mathbf{q})$ is time-reversal invariant, then there exists
a real structure\footnote{A subset $K$ of $\mathop{\rm
Herm}(3,\mathbb{C}\,)$ is said to be a real subset if there exists a
fixed $3 \times 3$ unitary matrix $U$ such that $U^*\,K\,U \subset
\mathop{\rm Herm}(3,\mathbb{R})$.} on $H_Z(N^\prime) \subset
\mathop{\rm Herm}(3,\mathbb{C}\,)$.  Therefore, without loss of
generality, we may assume that $H_Z(N^\prime)$ is a subset of
$\mathop{\rm Herm}(3,\mathbb{R})$, the
6\nolinebreak\mbox{-}dimensional vector subspace of all $3 \times 3$
symmetric matrices with real entries.  As we have omitted spin from
the molecular Hamiltonian $H_{\rm mol}$, the electronic Hamiltonian
$H_{\rm el}(\mathbf{q})$ is time-reversal invariant, provided that
there is no external magnetic field.

The Jahn-Teller Hamiltonian $H_{\rm JT}$ is defined as the traceless
part of $H_Z(\mathbf{q})$ by
\begin{equation}
  H_Z(\mathbf{q}) = s_0(\mathbf{q}) \, I + H_{\rm JT}(\mathbf{q}) \,,
\end{equation}
where $s_0(\mathbf{q}) =
\frac{1}{3}\trace\left(H_Z\left(\mathbf{q}\right)\right)$ and $I$ is the
$3 \times 3$ identity matrix.  The functions $s_0(\mathbf{q})$ and
$H_{\rm JT}(\mathbf{q})$ are usually approximated by writing them as a
polynomials in the components of $\mathbf{q}$.

For the $T \otimes (e \oplus t_2)$ and $T \otimes (e_g \oplus t_{2g})$
Jahn-Teller effects, we shall consider the coordinate vector
$(q_1\,,\,q_2)$ to transform according to the $e$ or $e_g$
representation and the $(q_3\,,\,q_5\,,\,q_5)$ to transform according
to the $t_2$ or $t_{2g}$ representation.  In the linear approximation
to the Jahn-Teller effect, the Hamiltonian is taken as
\begin{equation}
{\addtolength{\arraycolsep}{0.3em}
\renewcommand{\arraystretch}{1.2}
  H_{\rm JT}(\mathbf{q}) =
  \left[\begin{array}{ccc}
  \frac{1}{\sqrt{6}}\kappa_1q_1 - \frac{1}{\sqrt{2}}\kappa_1q_2 &
    -\frac{1}{\sqrt{2}}\kappa_2q_5 & -\frac{1}{\sqrt{2}}\kappa_2q_4 \\
  -\frac{1}{\sqrt{2}}\kappa_2q_5 & \frac{1}{\sqrt{6}}\kappa_1q_1 +
    \frac{1}{\sqrt{2}}\kappa_1q_2 & -\frac{1}{\sqrt{2}}\kappa_2q_3 \\
  -\frac{1}{\sqrt{2}}\kappa_2q_4 & -\frac{1}{\sqrt{2}}\kappa_2q_3 &
    -\sqrt{\frac{2}{3}}\,\kappa_1q_1 \label{modes1}
  \end{array}\right]}
\end{equation}
where $\kappa_1$ is the Jahn-Teller coupling constant for the $e$ or
$e_g$ normal mode and $\kappa_2$ is the Jahn-Teller coupling constant
to the $t_2$ or $t_{2g}$ normal mode.  The quadratic restoring term is
\begin{equation}
  s_0(\mathbf{q}) = \frac{1}{2}\beta_1\left(q_1^2 + q_2^2\right) 
                 + \frac{1}{2}\beta_2\left( q_3^2 + q_4^2 + q_5^2 \right),
\end{equation}
where $\beta_1$ is the quadratic coupling constant for the $e$ or $e_g$
normal mode and $\beta_2$ is the quadratic coupling constant for the
$t_2$ or $t_{2g}$ normal mode.

For the icosahedral $T \otimes h_g$ Jahn-Teller effect, the entire
5\nolinebreak\mbox{-}dimensional normal mode vector transforms
according to the $h_g$ representation.  In this case there is one
Jahn-Teller coupling constant $\kappa$ and only one quadratic coupling
constant $\beta$.  The linear Jahn-Teller Hamiltonian is
\begin{equation}
{\addtolength{\arraycolsep}{0.3em}
\renewcommand{\arraystretch}{1.2}
  H_{\rm JT}(\mathbf{q}) = \kappa\,
   \left[\begin{array}{ccc}
     \frac{1}{\sqrt{6}}q_1 - \frac{1}{\sqrt{2}}q_2 & -\frac{1}{\sqrt{2}}q_5 
        & -\frac{1}{\sqrt{2}}q_4 \\
     -\frac{1}{\sqrt{2}}q_5 & \frac{1}{\sqrt{6}}q_1 + \frac{1}{\sqrt{2}}q_2
        & -\frac{1}{\sqrt{2}}q_3 \\
     -\frac{1}{\sqrt{2}}q_4 & -\frac{1}{\sqrt{2}}q_3 & 
        - \sqrt{\frac{2}{3}}\,q_1
    \end{array}\right] \label{modes2}}
\end{equation}
and the quadratic restoring term in the Hamiltonian is
\begin{equation}
  s_0(\mathbf{q}) = \frac{1}{2}\beta \left(q_1^2 + q_2^2 + q_3^2 
                 + q_4^2 + q_5^2 \right)\,.
\end{equation}

We will view $H_{\rm JT}$ as a continuous map from $N^\prime$ to
$\mathop{\rm Herm}_0(3,\mathbb{R})$, the
\linebreak[4]
5\nolinebreak\mbox{-}dimensional subspace of traceless matrices in
$\mathop{\rm Herm}(3,\mathbb{R})$.  In this context, it is
straightforward to verify that both (\ref{modes1}) and (\ref{modes2})
are vector space isomorphisms.  If we define $\mathop{\rm
Herm}_0(3,\mathbb{R})$ as a metric space by endowing it with the metric
product
\begin{equation}
  \langle A,B \rangle = \kappa^{-2}\,\trace \left(AB\right)\,, \label{hs}
\end{equation}
then the map $H_{\rm JT}$ defined in (\ref{modes2}) is an isometry.
Therefore, it is reasonable for us to simply identify $N^\prime$ with
the corresponding open neighbourhood of the zero matrix in
$\mathop{\rm Herm}_0(3,\mathbb{R})$.  Note that the inner product
in~\ref{hs} is proportional to the usual Hilbert-Schmidt metric on
$\mathop{\rm Herm}_0(3,\mathbb{R})$.

By making a judicious choice of another metric on $\mathop{\rm
Herm}_0(3,\mathbb{R})$, it is also possible to arrange for
(\ref{modes1}) to be an isometry.  However, if we consider the equal
coupling cases of the $T \otimes \left(e \oplus t_2\right)$ and $T
\otimes \left( e_g \oplus t_{2g}\right)$ Jahn-Teller effects, then the
Hamiltonian (\ref{modes1}) simplifies to (\ref{modes2}) with $\kappa_1
= \kappa_2 = \kappa$ and $H_{\rm JT}$ is an isometry with the metric
defined in product~(\ref{hs}).  We shall assume for the remainder of
this paper that the linear Jahn-Teller Hamiltonian is given by
(\ref{modes2}), because the simplicity of the inner product
in~(\ref{hs}) allows for a cleaner explanation of our methods.
Nevertheless, our methods can be applied to more general inner product
structures on $\mathop{\rm Herm}_0(3,\mathbb{R})$, which allows the
Jahn-Teller effect for the Hamiltonian (\ref{modes1}) to be
considered.

Let $\mu_1$, $\mu_2$ and $\mu_3$ denote the eigenvalues of $H_{\rm
JT}\in \mathop{\rm Herm}_0(3,\mathbb{R})$.  Note that $\mu_i =
\lambda_i - \frac{1}{3}\sum_{i=1}^3 \lambda_i$, for $i = 1,2,3$.  We
define within $\mathop{\rm Herm}_0(3,\mathbb{R})$ the three maximal
regions $W_i$ for which the eigenvalue $\mu_i$ is nondegenerate.  We
call these regions similar degeneracy regions.  For example, $W_1
\subset \mathop{\rm Herm}_0(3,\mathbb{R})$ consists of all matrices
for which the bottom eigenvalue $\mu_1$ is nondegenerate.  There are
no constraints on $\mu_2$ and $\mu_3$, other than the imposed ordering
\begin{equation}
\mu_1(\mathbf{q}) < \mu_2(\mathbf{q}) \le \mu_3(\mathbf{q})
\end{equation}
of the eigenvalues and the condition that the $\sum_{i=1}^3 \mu_i$
vanish.  From their definition, it is clear that each $W_i$ is a
smooth open submanifold of $\mathop{\rm Herm}_0(3,\mathbb{R})$.  
In the following proposition, we determine the homotopy types of these
similar degeneracy regions by constructing deformation retractions.

\begin{proposition}
There exist strong deformation retractions
\begin{subequations}
\begin{align}
  r_1 :\, W_1 &\longrightarrow R_1\:,\label{ret1}\\
  r_2 :\, W_2 &\longrightarrow R_2 \:,\label{ret2}\\
  r_3 :\, W_3 &\longrightarrow R_3\:. \label{ret3}
\end{align}
\end{subequations}
The subspaces $R_1$ and $R_3$ have the homotopy type of
$\mathbb{R}\mathrm{P}(2)$, the 2-dimensional real projective space of
1\nolinebreak\mbox{-}dimensional vector subspaces in $\mathbb{R}^3$.
The subspace $R_2$ has the homotopy type of
3\nolinebreak\mbox{-}dimensional flag manifold
$\mathbb{R}\mathrm{F}(1,1,1)$, consisting of ordered triples of
mutually orthogonal 1\nolinebreak\mbox{-}dimensional vector subspaces
in $\mathbb{R}^3$.
\label{retract-prop}
\end{proposition}
\begin{proof}
We first construct these deformation retractions on the subspace of
diagonal matrices $W^D_i$ in $W_i$.  To this end, introduce the
coordinates
\begin{equation}
  b = \frac{\mu_2 - \mu_1}{\mu_3 - \mu_1}
  \qquad \text{and} \qquad
  r = \sqrt{\mu_1^2 + \mu_2^2 + \mu_3^2}
\label{bary}
\end{equation}
and define
\begin{equation}
  D(b,r) = \frac{r}{\sqrt{6(1-b+b^2)}}
       \begin{bmatrix}
        -(1+b) & 0 & 0 \\
       0 & 2b-1 & 0 \\
       0 & 0 & 2-b
       \end{bmatrix}
  \label{diag-matrix}
\end{equation}
Note that this is a well defined coordinate system for the diagonal
matrices in the similar degeneracy regions, because $\mu_3 > \mu_1$ in
each of the regions.  For each $W_i$, we now define homotopies
$\phi_i\;{:}\;W^D_i \times [0,1] \longrightarrow W^D_i$ as follows:
\begin{subequations}
\begin{align}
  \phi_1(D(b,r),t) & = D(b - t(b-1)\,,\, r -t (r-1)) \\
  \phi_2(D(b,r),t) & = D(b - t(b - \frac{1}{2})\,,\,r -
                      t(r-1)) \\
  \phi_3(D(b,r),t) & = D(b - tb\,,\,r - t(r-1))
\end{align}
\end{subequations}
These homotopies are deformations of $W^D_i$ onto
\begin{equation}
\begin{split}
  \phi_1(D,1) =
  \begin{bmatrix}
     -\sqrt{\frac{2}{3}} &        0           &    0 \\
             0           &  \frac{1}{\sqrt{6}} &    0 \\
             0           &                    & \frac{1}{\sqrt{6}}
  \end{bmatrix},
  \qquad
  \phi_2(D,1) =
  \begin{bmatrix}
         -\frac{1}{\sqrt{2}} &  0  &    0  \\
                  0          &  0  &    0  \\
                  0          &  0  &  \frac{1}{\sqrt{2}}
  \end{bmatrix}, \\
  \text{and}\qquad
  \phi_3(D,1) =
  \begin{bmatrix}
      -\frac{1}{\sqrt{6}}    &    0    &    0   \\
             0               & -\frac{1}{\sqrt{6}} & 0 \\
             0               &    0    &  \sqrt{\frac{2}{3}}
  \end{bmatrix}
\label{diag}
\end{split}
\end{equation}

The finite-dimensional spectral theorem or principal axis theorem
implies that all matrices $A \in W_i$ may be written as $A = O\,D\,O^t$,
where $O$ is an element of $\Orth(3)$, the group of orthogonal $3 \times 3$
matrices.  This allows us to extend the deformations $\phi_i$ to
deformations $\phi_i$ of $W_i$.  Specifically, for $A = O\,D\,O^t$, we
define
\begin{equation}
  \phi_i(A,t) = O\,\phi_i(D,t)\,O^t
\end{equation}
Note that $O$ is not uniquely determined, but any such $O$ will give
the same result for $\phi_i(A,t)$, because the isotropy subgroup of
the $\Orth(3)$ action is the same for all $t \in [0,1]$.  Thus, $\phi_i$
is a well defined deformation.  Furthermore, it is equiavariant with
respect to the action of $\Orth(3)$ on $\Herm_0(3,\mathbb{R})$.
The retractions $r_i$ are now defined by $r_i(A) = \phi_i(A,1)$.  The
homotopy types of the retracts $R_i$ follow from the
isotropy subgroups of the $\Orth(3)$ action on the diagonal
matrices in~(\ref{diag}).
\end{proof}

If $Z_1$ is the 1\nolinebreak\mbox{-}dimensional eigenspace associated
with $\mu_1$, then we may construct a real line bundle with fibre
$Z_1$ over $W_1$ and total space
\begin{equation}
  X_1 = \{ (v,A) \in \mathbb{R}^3 \times W_1  \mid
           A v = \mu_1 v \} \,.
\end{equation}
We denote this real line bundle by
\begin{subequations}
\begin{gather}
\parbox{\textwidth}{%
  \xymatrix{\xi_1:\,\mathbb{R}\; \ar @{^{(}->} [r] & X_1 \ar[d]^{\wp_1} \\
                                           & W_1
           }
}
\intertext{where the projection operator is}
  \wp_1\,{:}\; (v,A) \longmapsto A \, .
\end{gather}
\end{subequations}
Proving that $\xi_1$ satisfies the local triviality requirement of
a vector bundle is not difficult.  Similarly, the eigenspaces $Z_2$ and
$Z_3$ for $\mu_2$ and $\mu_3$, respectively, may be used to construct
real line bundles
\begin{equation}
\parbox{\textwidth}{%
  \xymatrix{ \xi_2:\,\mathbb{R}\; \ar @{^{(}->} [r] & X_2 \ar[d]^{\wp_2} \\
                                           & W_2
           }
}
\end{equation}
and
\begin{equation}
\parbox{\textwidth}{%
  \xymatrix{\xi_3:\, \mathbb{R}\; \ar @{^{(}->} [r] & X_3 \ar[d]^{\wp_3} \\
                                           & W_3
           }
}
\end{equation}

\begin{proposition}
Over $\mathbb{R}{\rm P}(2)$ is defined the canonical or tautological
line bundle, $\tau$, and the bundles $\xi_1$ and $\xi_3$ are
isomorphic to the pullback bundles $r_1^*\tau$ and $r_3^*\tau$,
respectively.  Over $\mathbb{R}\mathrm{F} (1,1,1)$ are defined three
canonical line bundles, $\tau_1$, $\tau_2$, and $\tau_3$, which are
associated with the first, second, and third elements, respectively,
in the triples of mutually orthogonal lines in $\mathbb{R}^3$.  All
three of these bundles are nontrivial, although their Whitney sum is
isomorphic to the trivial $\mathbb{R}^3$ vector bundle over
$\mathbb{R}\mathrm{F} (1,1,1)$.  The line bundle $\xi_2$ is isomorphic
to the pullback bundle $r_2^*\tau_2$.
\end{proposition}
\begin{proof}
This proposition follows from the definition of the tautological
bundles and the explicit construction of the line bundles, $\xi_i$, in
terms of eigenspaces.
\end{proof}

We have established that none of the line bundles $\xi_1$, $\xi_2$,
and $\xi_3$ are nontrivial,\footnote{We remark that not only are each
of the real line bundles $\xi_i$ nontrivial, but their
complexifications are nontrivial complex line bundles.  This contrasts
with the $E \otimes e$ Jahn-Teller effect in which the the electronic
eigenspaces form vector bundles which are isomorphic to the M\"obius
band.  The complexification of the M\"obius band is a trivial complex
line bundle, which allows the Born-Oppenheimer Hamiltonian for the $E
\otimes e$ Jahn-Teller effect to be written as a differential operator
acting on complex-valued functions rather than sections of a
nontrivial real line bundle.  However, this approach to the $E \otimes e$
Jahn-Teller effect cannot be applied to the Jahn-Teller effect for
orbital triplets.} because the maps $r_1$, $r_2$, and $r_3$ are
deformation retractions.  Furthermore, the first Stiefel-Whitney
classes of $\xi_1$ and $\xi_3$ are the nontrivial elements in the
cohomology groups
\begin{equation}
  H^1(W_1\,;\,\mathbb{Z}_2) = \mathbb{Z}_2  \qquad \text{and} \qquad 
  H^1(W_3\,;\,\mathbb{Z}_2) = \mathbb{Z}_2 \: ,
\end{equation}
respectively.  The first Stiefel-Whitney class of $\xi_2$ is one of
the three nontrivial elements of
\begin{equation}
  H^1(W_2\,;\,\mathbb{Z}_2) \cong \mathbb{Z}_2 \oplus \mathbb{Z}_2 \: .
\end{equation}
Precisely, it is the image of the first Stiefel-Whitney class of
$\tau_2$ under $r_2^*$.

The usual Born-Oppenheimer approximation cannot be used to calculate
the pseudorotational eigenvalues associated with each of the
electronic eigenvalues $\lambda_1$, $\lambda_2$, and $\lambda_3$,
because their corresponding real line bundles $\xi_1$, $\xi_2$, and
$\xi_3$ are nontrivial.  Instead, we must use a generalisation of the
Born-Oppenheimer approximation to nontrivial vector bundles.  To this
end, we denote the intersection of $W_i \cap N^\prime$ by $N_i$ and
the restriction of $\xi_i$ to $N_i$ by $\eta_i$.  The regions $N_i$
will be referred to as the similar degeneracy regions in $N^\prime$.

It follows from the Jahn-Teller Theorem that the electronic eigenvalue
$\lambda_i(\mathbf{q})$ is minimised on a submanifold of $N_i$.
Therefore, the molecular wavefunctions for low-energy pseudorotations
associated with $\lambda_i$ are mostly supported on $N_i$.  This leads
us to make the approximation that these molecular wavefunctions lie in
the Hilbert space $L^2(N_1 \times \mathbb{R}^{3z};\,\mathbb{R})$.  By
a standard construction~\cite{RS80}, this Hilbert space is isomorphic
to $L^2(N_i;\, L^2(\mathbb{R}^{3z};\, \mathbb{R}))$.  In order to
define the Born-Oppenheimer Hamiltonian associated with the $i^{\rm
th}$ electronic eigenspace, we shall view this Hilbert space in terms
of the trivial vector bundle
\begin{equation}
\parbox{\textwidth}{%
     \xymatrix{%
     \epsilon_i:\, 
       L^2(\mathbb{R}^{3z};\,\mathbb{R})\; \ar @{^{(}->} [r] &
       N_i \times L^2(\mathbb{R}^{3z};\,\mathbb{R}) \ar[d] \\
     & N_i}
}
\end{equation}
Triviality of $\epsilon_i$ implies that $L^2(\epsilon_i)$, the Hilbert
space of square integrable sections of $\epsilon_i$, is canonically
isomorphic to $L^2(N_i;\, L^2(\mathbb{R}^{3z};\, \mathbb{R}))$.  In
light of these two isomorphism, the Hilbert space of molecular
wavefunctions can be viewed as $L^2(\epsilon_i)$ and the molecular
Hamiltonian $H_{\rm mol}$ is therefore an operator on
$L^2(\epsilon_i)$.

The line bundle $\eta_i$ is a sub-bundle of $\epsilon_i$ and the
Hilbert space $L^2(\eta_i)$ is a subspace of $L^2(\epsilon_i)$.  We
denote the projection operator onto $L^2(\eta_i)$ by $P_i$ and the
projection operator onto the orthogonal complement of $L^2(\eta_i)$ in
$L^2(\epsilon)$ by $Q_i$.  In terms of these vector bundles, the
Born-Oppenheimer approximation states that the molecular Hamiltonian
$H_{\rm mol}$ can be approximately restricted to $L^2(\eta_i)$.
Verifying this statement amounts to showing that $Q_i\,H_{\rm
mol}\,P_i$ is small in some suitable sense.  In the simplified context
of diatomic molecules the accuracy of this approximation is analysed
in~\cite{CDS81,GH87,GH88} and an expansion for higher order
corrections is derived in~\cite{GH87,GH88}.  The approximate
restriction of $H_{\rm mol}$ to $L^2(\eta_i)$ is defined by
\begin{equation}
  B_i = P_i\,H_{\rm mol}\,P_i\,:\; L^2(\eta_i) \longrightarrow
  L^2(\eta_i)\,.
\end{equation}
It is called the Born-Oppenheimer Hamiltonian.

Straightforward algebra establishes that $B_i$ is a second
order differential operator of the form
\begin{equation}
  B_i = - \triangle_i + V_i \,,
  \label{bo-ham}
\end{equation}
where the Laplacian $\triangle_i$ is defined on the Hilbert space
$L^2(\eta_i)$ of square-integrable sections of $\eta_i$ and the
effective potential $V_i$ is a bundle endomorphism of $\eta_i$.  The
Laplacian acting on sections of a vector bundle is defined as follows.

\begin{definition}[Connection Laplacian]
Consider a vector bundle $\xi$ over a Riemannian manifold
$\mathcal{M}$, with connection $\nabla$.  The tangent and cotangent
bundles of $\mathcal{M}$ are denoted by $T\mathcal{M}$ and
$T^*\mathcal{M}$, respectively.  The vector space of smooth sections
of a vector bundle is denoted by $C^\infty(\,\cdot\,)$.  A connection
on $\xi$ is a bilinear map
\begin{equation}
  \nabla\,{:}\; C^\infty(T\mathcal{M}) \times C^\infty(\xi) \longrightarrow
     C^\infty(\xi)\,.
\end{equation}
This map, which is conventionally written as $\nabla\,{:}\; (X,\sigma)
\longmapsto \nabla_X\sigma$, satisfies:
\begin{list}{}{\setlength{\labelwidth}{4ex}\setlength{\labelsep}{1ex}
                          \setlength{\leftmargin}{5ex}}
\item[(i)]  $\nabla_X(f\sigma) = \left(X \cdot f\right)\sigma +
            f\nabla_X\sigma$, where $f$ is a function on $\mathcal{M}$ and $X
            \cdot f$ is the directional derivative of $f$ in the
            direction $X$.
\item[(ii)] $X \cdot \left<\sigma\,,\,\omega \right> =
            \left<\nabla_X\sigma\,,\, \omega\right> +
            \left<\sigma\,,\, \nabla_X\omega\right>$, where
            $\left<\,\cdot\,,\,\cdot\,\right>$ is the fibre-wise metric on
            $\xi$.
\end{list}
Then on $C^\infty(\xi)$, the covariant exterior differential relative
to $\nabla$,
\begin{equation}
  d\,{:}\; C^\infty(\xi) \longmapsto C^\infty(T^*\,\mathcal{M} \otimes \xi)\,,
\end{equation}
is defined by $d\,\sigma(X) = \nabla_X\sigma$.  Note that
$C^\infty(T^*\,\mathcal{M} \otimes \xi)$ is simply the space of $\xi$-valued
1-forms. 

The Connection Laplacian $\triangle\,{:}\;C^\infty(\xi) \longrightarrow
C^\infty(\xi)$ is defined by 
\begin{equation}
  \triangle\,\sigma = -d^*d\,\sigma\,, 
\end{equation}
where $d^*$ is the adjoint of $d$.  In terms of the connection, it is
given by $\triangle\,\sigma = \trace\nabla^2\sigma$.  Finally,
$\triangle$ is defined as a self-adjoint operator by a suitable
extension of its domain to a Sobolev subspace of $L^2(\xi)$.
\end{definition}

For the current situation in which $H_{\rm mol}$ is time-reversal
invariant and $\eta_i$ is a real line bundle, the connection
$\nabla_i$ is the flat connection on $\eta_i$.  Furthermore, a bundle
endomorphism of a line bundle may be canonically identified with
multiplication by a function on the base space.  Therefore, $V_i$ is
simply a potential function on $N_i$.  More generally, if $H_{\rm
mol}$ is not time-reversal invariant or if $\eta_i$ is a vector bundle
with fibre dimension greater than or equal to 2, then the differential
operator will contain derivatives of order one.  However, the leading
symbol of $B_i$ is the metric tensor of $N_i$ and therefore there
exists a unique connection $\nabla$ on $\eta_i$ for which
\begin{equation}
  B_i = -\triangle_i^\nabla + V_i\,,
\end{equation}
where $\triangle_i^\nabla$ is the Laplacian with respect to $\nabla$
and $V_i$ is a bundle endomorphism\linebreak[4] \mbox{\cite[Lemma~4.8.1]{G84}}.
For the case in which $N_i$ is a region in Euclidean space, the
connection $\nabla$ is the adiabatic connection described
in~\cite{S83}.

The simplest approximation for the effective potential $V_i$ is
\begin{equation}
  V_i(\mathbf{q}) = \lambda_i(\mathbf{q}) + V_{\rm
  nuc}(\mathbf{q})\label{pot}\,,
\end{equation}
where $V_{\rm nuc}(\mathbf{q})$ is the nuclei-nuclei potential defined
in (\ref{mol-ham}) and $\lambda_i$ is an eigenvalue of $H_Z$ defined
in~(\ref{Hz}).

It should be stressed that although the Hamiltonian (\ref{bo-ham})
looks locally like the usual Born-Oppenheimer Hamiltonian, the global
topology of a vector bundle is very important for any differential
operator acting on sections of it. As an example, consider that the
spectrum of the Laplacian acting on real-valued functions defined on
the unit circle $S^1$ is $\{n^2 \mid n = 0, 1, 2, \dots \}$, whereas
the spectrum of the Laplacian acting on sections of the M\"obius band
over $S^1$ is $\{n^2 \mid n = \frac{1}{2}, \frac{3}{2}, \frac{5}{2},
\dots \}$.

We have chosen to make the nuclear mass implicit in the definition of
$\triangle_i$ in (\ref{bo-ham}) by including it in the definition of
the Riemannian metric on $N_i$.  In the normal mode coordinate system
$(q_1,q_2,q_3,q_4,q_5)$, the covariant metric tensor on $N_i$ is
defined to be
\begin{equation}
  g_{kl} = 2\,\bar{M}\, \delta_{kl} \,,
\end{equation}
where $\delta_{kl}$ is the Kronecker delta tensor and $\bar{M}$ is the
effective nuclear mass associated with the $e_g$ and $t_{2g}$ normal
modes.  In terms of these coordinates, the Laplacian $\triangle_i$ has
the local form $\frac{1}{2\bar{M}}\sum_{k=1}^5
\frac{\partial^2}{\partial q_k^2}$.  Expressed in terms of the
Hilbert-Schmidt inner product on $\mathop{\rm Herm}_0(3,\mathbb{R})$
this Riemannian metric is
\begin{equation}
  g(A,B) = 2\,\kappa^{-2}\,\bar{M}\,\trace \left(AB\right)\,,
\end{equation}
for $A$ and $B$ in $\mathop{\rm Herm}_0(3,\mathbb{R})$.

The stable behaviour of the molecule is governed by the minima of the
potentials $V_i$, which are bounded from below.  The Jahn-Teller
theorem implies that $V_1$ cannot be minimised at $\mathbf{q} =
\mathbf{0}$.  Furthermore, the global minimum of $V_1$ is not an
isolated point, but is a critical manifold $Y$, which is a compact
submanifold of $N_1$.

Note that the points $\mathbf{q} \in N^\prime$ where two eigenvalues
$\lambda_i$ and $\lambda_j$ cross transversally are singular points
for the potentials $V_i$ and $V_j$, because the definition of these
potentials incorporates the ordering (\ref{order}).  The potentials
$V_2$ and $V_3$ may well be minimised at such singular points in
$N^\prime$ and the Jahn-Teller theorem tells us nothing about such
minima.  However, depending on the molecular Hamiltonian and the
details of the approximation used to derive $V_i$, each of the
effective potentials may also have differentiable local minima which
correspond to some stable vibronic excitation of the molecule.  The
Jahn-Teller theorem does apply in this case and we can conclude that
such minima cannot occur at $\mathbf{q} = \mathbf{0}$.
We will let $Y$ denote any submanifold of $N^\prime$ on which one of
the potentials $V_i$ has a nonsingular minimum.  The effective
potential or Born-Oppenheimer Hamiltonian under consideration will be
clear from context.

\section{Foliations}
\label{Sect_Geom}

In order to compute pseudorotational energy levels, we require
detailed information about the geometry of $Y$.  We shall identify $Y$
as a leaf of a foliation of a 4\nolinebreak\mbox{-}dimensional sphere
in the normal mode space $N^\prime$.  As a consequence, we will be
able to not only describe the geometry of $Y$, but also precisely how
$Y$ is embedded in this sphere.

A foliation is a decomposition of a manifold into a layering of
leaves.  Any point in the manifold must be in exactly one leaf, which
means that there can be no intersections between leaves and that the
manifold must be completely filled by the leaves.  We refer
to~\cite{M88} for the precise definition of a foliation and for an
introduction to foliation theory.

Let $S^4(r)$ denote the 4\nolinebreak\mbox{-}dimensional sphere of
radius $r$ centred at the origin in $N$.  Under the identification of
$N$ with $\mathop{\rm Herm}_0(3,\mathbb{R})$ we can view $S^4(r)$ as
being the spere of radius $r$ in $\mathop{\rm Herm}_0(3,\mathbb{R})$.
In this context, the sphere $S^4(r)$ is invariant under the adjoint
action of orthogonal group $\Orth(3)$.  This action determines a
decomposition of $S^4(r)$ into orbits, where the orbit containing a
matrix $A \in S^4(r)$ is
\begin{equation}
  \left\{B \in S^4(r) \mid B = O\,A\,O^t\,,\;\text{for 
  some}\; O \in \Orth(3)\right\}\,.
\end{equation}
The coordinate $b$ introduced in~(\ref{bary}) parameterises the orbit
space of this action.

Note that the potential $V_i$ depends on $H_{\rm el}$ only through the
eigenvalue $\lambda_i$.  Both $\lambda_i$ and the quadratic term
$s_0(\mathbf{q})$ are invariant under the adjoint action of $\Orth(3)$.
Therefore, the minimising submanifold $Y$ is entirely contained in a
sphere $S^4(r_0)$, where the radius $r_0$ depends on the physical
parameters $\kappa$, $\beta$, and $\bar{M}$.  The similar degeneracy
regions $N_1$, $N_2$, and $N_3$ each intersect this sphere $S^4(r_0)$
and we denote the similar degeneracy regions in $S^4(r_0)$ by
\begin{equation}
  \mathfrak{W}_i = N_i \cap S^4(r_0)\,,\quad
                \text{for $i=1,2,3$.}
\end{equation}
The hierarchy of the submanifolds $Y$, $S^4({r_0})$, and
$N^\prime$ are shown in Figure~\ref{fig1}.  Note that $\mathfrak{W}_2
= \mathfrak{W}_1 \cap \mathfrak{W}_3$ and that $\mathfrak{W}_2$ is the
submanifold in $S^4(r_0)$ where the three eigenvalues $\mu_1$,
$\mu_2$ and $\mu_3$ are mutually distinct.  Therefore, generically the
minimising submanifold $Y$ is contained in $\mathfrak{W}_2$.

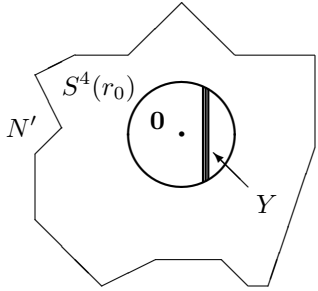
\begin{figure}\centering
  \begin{picture}(140,150)
    \thicklines
    \put(80,97.32){\line(0,-1){34.64}}
    \put(79,97.86){\line(0,-1){35.72}}
    \put(78,98.33){\line(0,-1){36.66}}
    \put(70,80){\circle{40}}
    \thinlines
    \put(70,80){\circle*{2}}
    \put(58,82){$\mathbf{0}$}
    \put(70,130){\line(-1,-1){20}}
    \put(70,130){\line(1,-1){20}}
    \put(50,110){\line(-1,0){20}}
    \put(90,110){\line(1,0){30}}
    \put(30,110){\line(-2,-1){15}}
    \put(120,110){\line(0,-1){35}}
    \put(120,75){\line(-1,-3){17.5}}
    \put(15,102.5){\line(1,-2){10}}
    \put(25,82.5){\line(-1,-1){10}}
    \put(15,72.5){\line(0,-1){25}}
    \put(15,47.5){\line(1,-1){25}}
    \put(40,22.5){\line(2,1){20}}
    \put(60,32.5){\line(1,0){25}}
    \put(85,32.5){\line(1,-1){10}}
    \put(95,22.5){\line(1,0){7.5}}
    \put(25,95){$\displaystyle S^4(r_0)$}
    \put(4,79){$\displaystyle N^\prime$}
    \put(95,60){\vector(-1,1){13}}
    \put(98,50){$\displaystyle Y$}
  \end{picture}
  \medskip

\caption{The open 5-dimensional manifold $N^\prime$ is starlike about
         $\mathbf{0}$.  The 4-dimensional sphere $S^4(r_0)$, which has
         radius $r_0$, is centred at $\mathbf{0}$ and is
         a submanifold of $N^\prime$.  The 3-dimensional submanifold $Y
         \subset S^4(r_0)$ is generically diffeomorphic to the real
         short flag manifold $\mathbb{R}\mathrm{F} (1,1,1)$.  \label{fig1}}
\end{figure}

The restriction of the deformation retractions (\ref{ret1}) and
(\ref{ret3}) to $S^4(r_0)$ are also deformation retractions.
Although, it is not generally true that the restriction of a
deformation retraction is also a deformation retraction, in this case
each deformation retraction $r_i$ may be written as a composition of a
deformation retraction of $W_i$ to $\mathfrak{W}_i$ with a further
retraction of $\mathfrak{W}_i$.  We denote $r_1(\mathfrak{W}_1)$ by
$\mathfrak{R}_1$ and $r_3(\mathfrak{W}_3)$ by $\mathfrak{R}_3$.  It
follows from Proposition~\ref{retract-prop} that both $\mathfrak{R}_1$
and $\mathfrak{R}_3$ are diffeomorphic to the real projective space
$\mathbb{R}{\rm P}(2)$.  However, the restriction of the deformation
retraction (\ref{ret2}) to $S^4(r_0)$ implies that the similar
degeneracy region $\mathfrak{W}_2$ is diffeomorphic to
$\mathbb{R}\mathrm{F}(1,1,1) \times (0,1)$.  This geometrical
dichotomy suggests that the problem of calculating the
pseudorotational spectra also divides into two cases.  The
pseudorotational spectra corresponding to $\lambda_1$ and $\lambda_3$
will be referred to as type~I spectra, whereas the spectrum
corresponding to $\lambda_2$ will be referred to as the type~II
spectrum.  This terminology will also be used for the similar
degeneracy regions associated with each of these eigenvalues.

The similar degeneracy regions $\mathfrak{W}_i$ are open Riemannian
submanifolds of the\linebreak[4] 4\nolinebreak\mbox{-}dimensional
sphere $S^4(r_0) \subset N$.  They obtain their Riemannian metric from
the restriction of the metric $g$ on $N$.  It is important to note
that the spectrum of the Born-Oppenheimer Hamiltonian $B_i$ depends
crucially on the Riemannian geometry of $\mathfrak{W}_i$.  The first
observation that we can make about the geometry of $\mathfrak{W}_1$
and $\mathfrak{W}_3$ is that they are isometric to each other as
Riemannian manifolds.  This follows from the inherent symmetry of the
metric $g$.

For $0 < b < 1$, the orbits of the adjoint $\Orth(3)$ action on $S^4(r_0)$
are diffeomorphic to the real flag manifold
$\mathbb{R}\mathrm{F}(1,1,1)$.  These orbits, which we denote by
$\mathfrak{F}(b)$, foliate the similar degeneracy region
$\mathfrak{W}_2$.  The exceptional orbits when $b = 0$ and $b=1$ are
$\mathfrak{R}_1$ and $\mathfrak{R}_3$, respectively, because the
deformation retractions $r_1$ and $r_3$ are equivariant with respect
to the $\Orth(3)$ action.  Furthermore, $\mathfrak{R}_1$ and
$\mathfrak{R}_3$ are isometric Riemannian manifolds.  The
decomposition of $S^4(r_0)$ into similar degeneracy regions, which are
in turn foliated into the orbits $\mathfrak{F}(b)$, is shown in
Figure~\ref{fig2}.

\begin{figure}\centering
  \begin{picture}(240,168)
    \thicklines
    \put(30,100){\circle*{12}}
    \put(210,100){\circle*{12}}
    \put(40,101.5){\line(1,0){160}}
    \put(40,100.8){\line(1,0){160}}
    \put(40,100.1){\line(1,0){160}}
    \put(40,99.3){\line(1,0){160}}
    \put(40,98.5){\line(1,0){160}}
    \put(130,85){\line(0,1){30}}
    \put(135,120){\makebox(0,0)[l]{$\mathfrak{F}(b)\approx \mathbb{R}\mathrm{F}(1,1,1)$}}
    \put(17,100){\makebox(0,0)[r]{\parbox{6em}{%
            \begin{align*}
               \mathfrak{R}_1 & = \mathfrak{F}(0) \\
                              & \approx \mathbb{R}{\rm P}(2)
            \end{align*}}}}
    \put(224,100){\makebox(0,0)[l]{\parbox{6em}{%
             \begin{align*}
                \mathfrak{R}_3 & = \mathfrak{F}(1) \\
                               & \approx \mathbb{R}{\rm P}(2)
             \end{align*}}}}
    \put(111.5,60){\makebox(0,0){$\underbrace{\hspace*{175pt}}$}}
    \put(111.5,51){\makebox(0,0)[t]{$\mathfrak{W}_1$}}
    \put(128.5,30){\makebox(0,0){$\underbrace{\hspace*{175pt}}$}}
    \put(128.5,21){\makebox(0,0)[t]{$\mathfrak{W}_3$}}
    \put(120,140){\makebox(0,0){$\overbrace{\hspace*{157pt}}$}}
    \put(120,150){\makebox(0,0)[b]{$\mathfrak{W}_2 \approx
                  \mathbb{R}\mathrm{F}(1,1,1) \times (0,1)$}}
  \end{picture}
\caption{Decomposition of $S^4(r_0)$
\label{fig2}}
\end{figure}
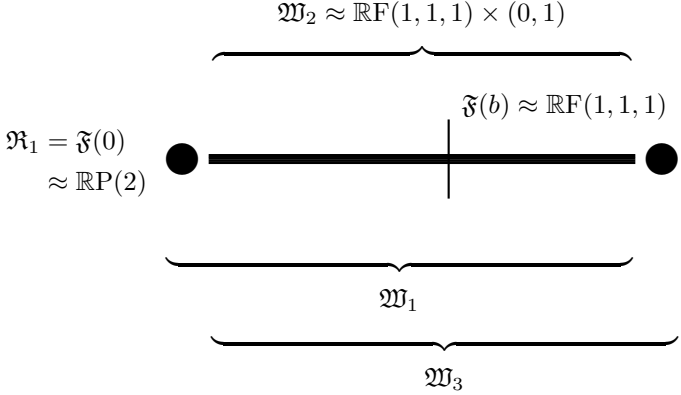

Using a standard method for defining a coordinate system on a
homogeneous space \cite[Section~II.4]{H78}, we can construct
coordinate systems on $\mathfrak{F}(b)$ in terms of the Lie algebra of
$\Orth(3)$.  Specifically, define the diagonal matrix
\begin{subequations}
\begin{gather}
  D(b) = \begin{bmatrix}
       \frac{-r_0(1+b)}{(\sqrt{6})\,(\sqrt{1-b+b^2})} & 0 & 0 \\
       0 & \frac{r_0(2b-1)}{(\sqrt{6})\,(\sqrt{1-b+b^2})} & 0 \\
       0 & 0 & \frac{r_0(2-b)}{(\sqrt{6})\,(\sqrt{1-b+b^2})}
       \end{bmatrix}
\intertext{and the following basis for the Lie algebra of $\Orth(3)$:} 
 E_1 = \begin{bmatrix}
         0 & 0 & 0 \\
         0 & 0 & -1\\
         0 & 1 & 0
        \end{bmatrix} 
  \qquad
 E_2 = \begin{bmatrix}
        0 & 0 & -1\\
        0 & 0 & 0 \\
        1 & 0 & 0
        \end{bmatrix}
  \qquad
 E_3 = \begin{bmatrix}
        0 & -1 & 0 \\
        1 & 0 & 0  \\
        0 & 0 & 0
        \end{bmatrix}
\end{gather}
Then, terms of these matrices, a coordinate system for $S^4(r_0)$ is
\begin{equation}
  \begin{split}
  A(b,x_1,x_2,x_3) = \left[\exp(x_2E_2)\right]^t\,\left[\exp(x_1E_1)\right]^t\,
                     \left[\exp(x_3E_3)\right]^t\,D(b)\, \\
                     \left[\exp(x_3E_3)\right]\,\left[\exp(x_1E_1)\right]\,
                     \left[\exp(x_2E_2)\right]
  \end{split}
\end{equation}
\label{W1-coords}
\end{subequations}
Using this coordinate system and related coordinate systems, it is
straightforward to compute various geometrical quantities.

First, the Gaussian curvatures of $\mathfrak{R}_1$ and
$\mathfrak{R}_3$ are constant and both equal to
$\left(3r_0^{\>2}\right)^{-1}$.  It is interesting to note that for
$r_0 = 1$, the submanifolds $\mathfrak{R}_1$ and $\mathfrak{R}_3$ are
both isometric to the image of the well-known Veronese
embedding~\cite{BCO03,CR85} of $\mathbb{R}{\rm P}(2)$ into $S^4$.
This embedding has also been studied by Massey~\cite{M74} who called
it the ``canonical'' embedding of $\mathbb{R}{\rm P}(2)$ into a sphere
of minimum dimension.  The lowest dimensional sphere into which
$\mathbb{R}{\rm P}(2)$ can be embedded is $S^4$.

Consider now the generic leaves when $b \in (0,1)$.  The leaves
$\mathfrak{F}(b)$ form a geodesically parallel family of submanifolds
in $S^4(r_0)$.  In terms of the coordinate system~(\ref{W1-coords}),
the vector field $\frac{\partial}{\partial b}$ is a normal vector
field to $\mathfrak{F}(b)$, which generates a geodesic flow from leaf
to leaf.  Furthermore, the adjoint action of $\Orth(3)$ on
$\Herm_0(3,\mathbb{R})$ defines a connected closed subgroup of the
isometry group of $\Herm_0(3,\mathbb{R})$.  Therefore, it follows
from~\cite[Prop.~3.8.2]{BCO03} that for $b \in (0,1)$, the orbits of
this action form an isoparametric family of hypersurfaces in
$S^4(r_0)$\footnote{Isoparametric hypersurfaces in a manifold
$\mathcal{M}$ are generally defined as the regular level hypersurfaces
of an isoparametric function on $\mathcal{M}$.  However, an equivalent
definition for hypersurfaces in a sphere is that a family is
isoparametric if and only if the hypersurfaces are geodesically
parallel and have constant principal curvatures.  See
\cite[Sect.~5.2.e]{BCO03} and \cite[Sect.~3.1]{CR85}.  Such a family
is said to be of type $n$ if there are $n$ distinct principal
curvatures.}.  Hence, the principal curvatures of $\mathfrak{F}(b)$
are constant.  By direct computation, we find that they are
\begin{subequations}
\begin{align}
  k_1(b) & = \frac{1}{\sqrt{3}\,r_0}\> 
                   \frac{b - 2}{b}\\
  k_2(b) & = \frac{1}{\sqrt{3}\,r_0}\> 
                   \left(2b -1 \right)\\
  k_3(b) & = \frac{1}{\sqrt{3}\,r_0}\> 
                   \frac{1 + b}{1 - b}
\end{align}
\label{princ}
\end{subequations}
Note that these principal curvatures are distinct, satisfying $k_1(b)
< k_2(b) < k_3(b)$, for all $b \in \left( 0,1 \right)$, and therefore,
this is an isoparametric family of type~3.  For $r_0 = 1$, it is
isometric to Cartan's isoparametric family in $S^4$~\cite{C38,N75}.
The focal submanifolds of this family are $\mathfrak{R}_1$ and
$\mathfrak{R}_3$. 

An \mbox{$r$-dimensional} submanifold is said to be minimal if its
$r$-dimensional volume is locally extremal.  It follows from results
in~\cite{H66} that $\mathfrak{R}_1$ and $\mathfrak{R}_3$ are minimal
submanifolds of $S^4(r_0)$, because they are isolated orbits of the
adjoint action of $\Orth(3)$ on $S^4(r_0)$.  Also, at least one of the
submanifolds $\mathfrak{F}(b)$, for $b \in \left(0,1\right)$, must be
minimal.  The submanifold $\mathfrak{F}(b)$ will be minimal if and
only if its mean curvature vanishes.  The mean curvature of
$\mathfrak{F}(\beta)$ is
\begin{equation}
  h(b) = \frac{1}{3}\sum_{i = 1}^3k_i(b)
       = \frac{1}{3\sqrt{3}\,r_0}\>\frac{(2 - b)(2b
             - 1)(b + 1)}{(1 - b)b}
\label{mean-curv}
\end{equation}
and the only value of $b \in \left(0,1\right)$ for which $h(b)$ vanishes is $b =
\frac{1}{2}$.  Therefore, $\mathfrak{F}(\frac{1}{2})$ is a minimal
submanifold.

The Gauss-Codazzi equations imply that for hypersurfaces in $S^4$,
the scalar curvature is given by
\begin{equation}
  \mathop{\rm scal} = 6 - \sum_{i=1}^3k_i^2 + 9\,h^2\,.
\end{equation}
Substituting from (\ref{princ}), we find that the scalar curvature of
$\mathfrak{F}(b)$ is $\mathop{\rm scal}(b) = 0$, for all $b \in
(0,1)$.

The retractions $r_1$ and $r_3$ determine foliations of the similar
degeneracy regions $\mathfrak{W}_1$ and $\mathfrak{W}_3$,
respectively. These foliations will be useful in determining the
spectra of the Born-Oppenheimer Hamiltonians.  The leaves of the
foliation $\mathcal{F}_i$ of $\mathfrak{W}_i$ are defined to be
\begin{equation}
  \mathfrak{L}_i(y) = r_i^{-1}(y)\,,
\end{equation}
for each $y \in \mathfrak{R}_i$ and $i=1,3$. 
By construction there is a one-to-one correspondence between leaves of
$\mathcal{F}_i$ and points in $\mathfrak{R}_i$.  Hence, we can view
$\mathfrak{R}_i$ as the leaf space of $\mathcal{F}_i$.  Furthermore, 
the foliations $\mathcal{F}_1$ and $\mathcal{F}_3$ are isometric,
with the substitution $b \mapsto 1-b$ taking one to the other.

\begin{proposition}
\label{geo}
The foliations $\mathcal{F}_1$ and $\mathcal{F}_3$ are totally
geodesic foliations\footnote{A submanifold $\mathcal{M} \subset
\mathcal{M}^\prime$ is said to be totally geodesic if every geodesic in
$M$ is also a geodesic in $\mathcal{M}^\prime$.  A foliation is totally
geodesic if each leaf of the foliation is totally geodesic.}  of
$\mathfrak{W}_1$ and $\mathfrak{W}_3$, respectively.
\end{proposition}
\begin{proof}
It suffices to consider just $\mathcal{F}_1$, because the two
foliations are isometric.  The adjoint action of $\Orth(3)$ on
$\mathfrak{W}_1$ is a transitive action on the leaf space of
$\mathcal{F}_1$.  This implies that any two leaves of $\mathcal{F}_1$
are isometric and we need only show that one of the leaves is totally
geodesic.  

A submanifold is totally geodesic if and only if its second
fundamental form vanishes \mbox{\cite[Theorem~4.1]{O87}}.  In terms of
the coordinate system in~(\ref{W1-coords}), the leaf corresponding to
the diagonal matrix in the leaf space $\mathfrak{R}_1$ is
\begin{equation}
  \mathfrak{L}_1(D) = \{A(b,x_1,x_2,x_3)\mid b \in [0,1)\,,\,x_1=0\,,\,x_2=0\,,
        \,x_3 \in [0,\pi)\}\,.
\end{equation}
An orthogonal frame field for the tangent bundle $T\mathfrak{L}_1(D)$ is
\begin{equation}
  \frac{\partial}{\partial x_0} = \left .\frac{\partial A}{\partial b}
  \right|_{x_1=x_2=0} 
  \qquad\text{and}\qquad
  \frac{\partial}{\partial x_3} = \left .\frac{\partial A}{\partial
    x_3} \right|_{x_1=x_2=0} 
\end{equation}
In terms of this coordinate system, the second fundamental form of
$\mathfrak{L}_1(D)$ is
\begin{equation}
  \Psi\left(\frac{\partial}{x_i},\frac{\partial}{x_j}\right) =
  \sum_{l=1}^2 \Gamma_{ij}^l\,\frac{\partial}{\partial x_l}\,,
\end{equation}
where $i,j=0,3$ and $\Gamma_{ij}^l$ are the Christoffel symbols of the
second kind.  Computing the relevant Christoffel symbols gives
\begin{equation}
\begin{alignedat}{3}
  \Gamma_{00}^1 &= 0 &\qquad  \Gamma_{00}^2 &=0 \\
  \Gamma_{33}^1 &= 0 &\qquad \Gamma_{33}^2 &= 0 \\
  \Gamma_{03}^1 &= \Gamma_{30}^1 = 0 &\qquad \Gamma_{03}^2 &= 
    \Gamma_{30}^2 = 0
\end{alignedat}\label{christoffel}
\end{equation}
on $\mathfrak{L}_1(D)$. Therefore, $\Psi = 0$ on $\mathfrak{L}_1(D)$.
\end{proof}

In order to give a detailed description of the geometry of the leaves
of $\mathcal{F}_1$ and $\mathcal{F}_3$, we first consider $S^2(r_0)$,
the 2\nolinebreak\mbox{-}dimensional sphere with the same radius as
the ambient \mbox{$4$-sphere}, $S^4(r_0)$.  On $S^2(r_0)$ define
geographical coordinates $(\varphi\,,\,\vartheta)$, where $\varphi \in
[0,2\pi)$ is the angle of longitude and $\vartheta \in
[-\frac{\pi}{2},\frac{\pi}{2}]$ is the angle of latitude.  A
straightforward computation using the coordinate system
in~(\ref{W1-coords}) establishes the following result.

\begin{proposition}
For $i=1,2$, each leaf $\mathfrak{L}_i(y)$ of $\mathcal{F}_i$ is
isometric to the submanifold of $S^2(r_0)$ which is defined by
\begin{equation}
  \left\{(\varphi\,,\,\vartheta) \in S^2(r_0) \mid
  \frac{\pi}{6} < \vartheta \le \frac{\pi}{2}\right\}.
\end{equation}
In this coordinate system, the point $y\in \mathfrak{R}_i$ corresponds
to the north pole at $\vartheta = \frac{\pi}{2}$.  In other words,
$\mathfrak{L}_i(y)$ is a 2\nolinebreak\mbox{-}dimensional sphere of
radius $r_0$, which has been truncated at $30^\circ$ latitude.
This truncated sphere is attached to $\mathfrak{R}_1$ at $y$, the north
pole of $\mathfrak{L}_i(y)$. 
\end{proposition}

On a leaf $\mathfrak{L}_1(y)$ of the foliation $\mathcal{F}_1$, the
coordinate $b$ is related to the latitude coordinate $\varphi$ by
\begin{equation}
  \cos\varphi = \frac{\sqrt{3}\,b}{2\sqrt{\mathstrut 1 - b + b^2}}\,.
  \label{latitude}
\end{equation}
It follows that the leaves of the foliation $\mathcal{F}_1$ intersect
$\mathfrak{F}(b)$ in circles of constant latitude on
$\mathfrak{L}_1(y)$.  This defines a foliation $\mathcal{C}_1(b)$ of
$\mathfrak{F}(b)$ by circles of radius $r_0 \, \cos\varphi$.  The
leaves of this foliation are $\mathfrak{C}_1(y) = \mathfrak{L}_1(y)
\cap \mathfrak{F}(b)$, for $y \in \mathfrak{R}_1$, which implies that
the leaf space of $\mathcal{C}_1(b)$ is homeomorphic to
$\mathfrak{R}_1$.  It follows from the Christoffel symbols calculated
in~(\ref{christoffel}) that the leaves of $\mathcal{C}_1(b)$ are
geodesics in $\mathfrak{F}(b)$, which means that $\mathcal{C}_1$ is a
geodesic flow.  Similarly, the intersection of the leaves of
$\mathcal{F}_3$ with $\mathfrak{F}(b)$ give another geodesic flow,
$\mathcal{C}_3(b)$. Furthermore, the main theorem in~\cite{Ep72} implies
that $\mathcal{C}_1(b)$ and $\mathcal{C}_3(b)$ are Seifert fibrations of
$\mathfrak{F(b)}$.

The leaves of the foliation $\mathcal{F}_1$ are also the fibres of the
normal bundle of the embedding of $\mathfrak{R}_1$ in $S^4(r_0)$.  We
shall stretch our notation by using $\mathcal{F}_1$ to also denote
this normal bundle.  To understand $\mathcal{F}_1$ as a bundle, it is
useful to compute its Euler class, which is an element the cohomology
group $H^2(\mathbb{R}{\rm P}(2)\,;\,\mathcal{Z})$.  It is necessary to
use cohomology with twisted integer coefficients, denoted by
$\mathcal{Z}$, because $\mathfrak{R}_1$ is a nonorientable manifold.

\begin{proposition}
The Euler class of $\mathcal{F}_1$ is $e(\mathcal{F}_1) = \pm 2 \in
H^2(\mathbb{R}{\rm P}(2)\,;\,\mathcal{Z}) \cong \mathbb{Z}$, where a
choice of sign corresponds to a choice of orientation on
$S^4(r_0)$.
\end{proposition}
\begin{proof}
This follows directly from a general result that was conjectured by
H.~Whitney and proven by W.~S.~Massey~\cite{M69}.  Specifically, if
$\mathcal{M}$ is a closed, connected, nonorientable surface embedded
in $S^4$, then the Euler class of the normal bundle has one of the
following values:
\begin{equation}
  2\chi-4\,,\,2\chi\,,\,2\chi+4\,,\,\dots\,,\,4-2\chi\,,
\end{equation}
where $\chi$ is the Euler characteristic of $\mathcal{M}$.

The Euler characteristic of $\mathfrak{R}_1$ is $\chi = \sum (-1)^q
\mathop{\rm dim}H^q(\mathfrak{R}_1\,;\,\mathbb{R}) = 1$.
\end{proof}

Some additional insight can be gained by a more direct computation of
$e(\mathcal{F}_1)$.  Consider the universal covering
\begin{equation}
\parbox{\textwidth}{%
  \xymatrix{\mathbb{Z}_2\; \ar @{^{(}->} [r] 
             & S^2(\sqrt{3}\,r_0) \ar[d]^{p_1} \\
             & \mathfrak{R}_1
           }
}	   
\end{equation}
The pullback by $p_1$ of the tangent bundle $T\mathfrak{R}_1$ is
isomorphic to the tangent bundle $TS^2(\sqrt{3}\,r_0)$, because $p_1$
is a covering map.  Therefore, the induced map on cohomology,
$p_1^*\,{:}\; H^2(\mathfrak{R}_1\,;\,\mathcal{Z}) \longrightarrow
H^2(S^2\,;\,\mathbb{Z})$, must map the Euler class of
$T\mathfrak{R}_1$ to the Euler class of $TS^2(\sqrt{3}\, r_0)$.  This
implies that the map $p_1^*$ above must be the multiplication by 2 map
from $\mathbb{Z}$ to $\mathbb{Z}$, because the Euler characteristic of
$S^2$ is 2.

Having indentified the map $p_1^*\,{:}\;
H^2(\mathfrak{R}_1;\mathcal{Z}) \longrightarrow H^2(S^2;\mathbb{Z})$,
we now consider the pullback bundle $p_1^*(\mathcal{F}_1)$.  By
looking directly at how $\mathcal{F}_1$ is constructed, it is apparent
that $p_1^*(\mathcal{F}_1)$ is isomorphic to the associated projective
tangent bundle over $PTS^2$.  This bundle is constructed by
considering the unit tangent bundle over $S^2$ and identifying each
unit tangent vector with its negative.  Therefore, the circle fibres
of the unit tangent bundle double cover the $\mathbb{R}\mathrm{P}(1)$
fibres of projective tangent bundle.  Of course, $PTS^2$ is also an
$S^1$ bundle over $S^2$.  Furthermore, the Euler class of this bundle
is $e(PTS^2) = 2\,e(TS^2) = 4$, which implies that the Euler class
$e(p_1^*(\mathcal{F}_1))$ is equal to $\pm 4$, where the sign depends
on whether our choice of orientation for $S^2$ is consistent with the
orientation on $S^4(r_0)$.  Recalling that we have already shown that
$p_1^*\,{:}\; H^2(\mathfrak{R}_1;\mathcal{Z}) \longrightarrow
H^2(S^2\,;\,\mathbb{Z})$ is multiplication by $2$, we conclude that
the Euler class of $\mathcal{F}_1$ is $\pm 2$.

The nontriviality of the Euler class $\mathfrak{e}(\mathcal{F}_1)$
implies that $\mathfrak{W}_1$ is not simply a product of the leaf space
$\mathfrak{R}_1$ with a generic leaf $\mathfrak{L}_1(y)$, but rather, the
leaves are twisted together in a complicated fashion.  Note that up to
sign, the Euler class $e(\mathcal{F}_3)$ is equal to the Euler
class of $\mathcal{F}_1$.

We now examine further the geodesic flows $\mathcal{C}_i(b)$ in
$\mathfrak{F}(b)$.  The restriction of the deformation retraction
$r_1$ to $\mathfrak{F}(b)$ is a submersion from $\mathfrak{F}(b)$ to
$\mathfrak{R}_1$.  The fibres of this submersion are the leaves of the
geodesic flow $\mathcal{C}_1$.  The manifold $\mathfrak{R}_1$ can be
viewed as the leaf space of $\mathcal{C}_1(b)$.  However, note that
although $\mathfrak{R}_1$ carries a natural Riemannian metric when
viewed as a submanifold of $S^4(r_0)$, this metric may not be the
natural metric for $\mathfrak{R}_1$ when viewed as the leaf space of
$\mathcal{C}_1(b)$.  The tangent space $T_x\mathfrak{F}(b)$ decomposes
into a 1\nolinebreak\mbox{-}dimensional vertical subspace
$T_x\mathcal{C}_1(b)$ and a 2\nolinebreak\mbox{-}dimensional
horizontal subspace $N_x\mathcal{C}_1(b)$.  This decomposition allows
us to decompose the measure on $\mathfrak{F}(b)$ into a leaf measure
and an invariant transverse measure.  With this decomposition we are
able to compute the volume of $\mathfrak{F}(b)$.

\begin{proposition}
The volume of $\mathfrak{F}(b)$ is
\begin{equation}
  \vol(\mathfrak{F}(b)) =
  \frac{6\sqrt{3}\,b(1-b)\pi^2 r_0^{\> 3}}{\left(1 - b +
  b^2\right)^{\frac{3}{2}}}\,. \label{vol}
\end{equation}
\end{proposition}
Recall from~(\ref{mean-curv}) that $\vol(\mathfrak{F}(b))$ must have a
local extremum at $b = \frac{1}{2}$.  From this proposition, we see
that $\vol(\mathfrak{F}(b))$ has a maximum at $b = \frac{1}{2}$, for
$b \in (0,1)$.
\begin{proof}
The principal curvatures in~(\ref{princ}) are eigenvalues of the
second fundamantal form of the embedding of $\mathfrak{F}(b)$ in
$S^4(r_0)$.  Eigenvectors corresponding to these eigenvalues are
called principal vector fields on $\mathfrak{F}(b)$.  In can be
verified by explicitly computing the second fundamental form that a
tangent vector field to the leaves of $\mathcal{C}_1(b)$ is a
principal vector field associated with the principal curvature
$k_1(b)$.  This implies that $\mathcal{C}_1(b)$ is a principal
foliation of $\mathfrak{F}(b)$.  Similarly, the foliation
$\mathcal{C}_3(b)$ is the principal foliation associated with the
principal curvature $k_3(b)$.  The remaining principal foliation,
which is associated with $k_2(b)$, shall be denoted by
$\mathcal{C}_2(b)$.  Although $\mathcal{C}_2(b)$ is also a circle
foliation of $\mathfrak{F}(b)$, it does not arise from a deformation
retraction in the same way as $\mathcal{C}_1$ and $\mathcal{C}_3$.

The horizontal subspaces define the normal bundle $N\mathcal{C}_1(b)$
and the tangent bundle $T\mathfrak{F}(b)$ decomposes into the Whitney
sum $T\mathcal{C}_1(b) \oplus N\mathcal{C}_1(b)$.  Furthermore,
$dr_1$, the differential of $r_1$, maps the normal bundle
$N\mathcal{C}_1(b)$ to the tangent bundle $T\mathfrak{R}_1$.  In terms
of the coordinate system~(\ref{W1-coords}), orthonormal principal
vector fields for $\mathcal{C}_2(b)$, and $\mathcal{C}_3(b)$,
respectively, are
\begin{subequations}
\begin{align}
  V_2 &= \frac{r_0\sqrt{1 -b + b^2}}{\sqrt{3}\,(1-b)}\left[
         \cos x_3 \frac{\partial}{\partial x_1} + \frac{\sin x_3}{\cos
         x_1} \frac{\partial}{\partial x_2} + \frac{\sin x_1 \sin
         x_3}{\cos x_1}\frac{\partial}{\partial x_3}\right] \\[1\jot]
  V_3 &= \frac{r_0\sqrt{1 - b + b^2}}{\sqrt{3}} \left[ -\sin x_3
         \frac{\partial}{\partial x_1} + \frac{\cos x_3}{\cos
         x_1}\frac{\partial}{\partial x_2} + \frac{\cos x_3 \sin x_1}{\cos
         x_1} \frac{\partial}{\partial x_3}\right]
\end{align}
\end{subequations}
These vector fields provide a frame for $N\mathcal{C}_1(b)$.  Also,
the vector fields
\begin{equation}
  X_1 = \frac{r_0}{\sqrt{3}}\left . \frac{\partial}{\partial
      x}\right|_{b=0}
  \quad\text{and}\quad
  X_2 = \frac{r_0}{\sqrt{3}\cos x_1}\left .\frac{\partial}{\partial
      x_2}\right|_{b=0}
\end{equation}
give an orthonormal frame for $\mathfrak{R}_1$.

It follows from the definition of $r_1$ as a deformation retract that
\begin{equation}
\begin{alignedat}{2}
  dr_1\left(\frac{\partial}{\partial x_1}\right) 
    &= \left .\frac{\partial}{\partial x_1}\right|_{b=0}
    &\qquad
  dr_1\left(\frac{\partial}{\partial x_2}\right) 
    &= \left .\frac{\partial}{\partial x_1}\right|_{b=0}\\[1\jot]
  dr_1\left(\frac{\partial}{\partial x_3}\right) &= 0
\end{alignedat}
\end{equation}
Therefore, the Jacobian of $r_1$ is $\frac{1-b+b^2}{1-b}$.  This
implies that $r_1$ gives a decomposition of the measure on
$\mathfrak{F}(b)$ into horizontal and vertical components if the leaf
space is taken to be $\mathbb{R}\mathrm{P}(2)$ with constant Gaussian
curvature
\begin{equation}
  K_1 = \frac{1-b}{3r_0^2(1 - b + b^2)}
\end{equation}
The volume of $\mathbb{R}\mathrm{P}(2)$ with curvature $K_1$ is
$\frac{6\pi r_0^2(1-b)}{1 - b + b^2}$ and from
equation~(\ref{latitude}) the circumference of $\mathcal{C}_1(b)$ is
$\frac{\sqrt{3}\pi r_0 b}{2\sqrt{1 - b + b^2}}$. The proposition then
follows from Fubini's theorem. 
\end{proof}

\section{Covering Spaces}
\label{Sect_CS}

Recall that the Born-Oppenheimer Hamiltonian, $B_i$, acts on $L^2$
sections of the line bundle $\eta_i$ over $N_2$.  Furthermore, the
effective potential, $V_i$, is minimised on a submanifold $Y \subset
S^4(r_0)$, which corresponds to one of the leaves $\mathfrak{F}(b)$ of
the isoparametric foliation of $S^4(r_0)$.  Molecular excitations may
be viewed as motion on the energy minimising submanifold $Y$ coupled
with oscillations normal to $Y$.  If the quadratic Jahn-Teller
coupling constant is large, then these two types of excitations will
be effectively decoupled.  In this case, eigensections of $B_i$ are
approximated by tensor products of sections in $\eta_i\vert_Y$,
the restriction of the line bundle $\eta_i$ to $Y$, with functions of
the coordinates for $N_2$ which are normal to $Y$.  These sections of
$\eta_i\vert_Y$ will be eigensections of the restriction of $B_i$ to
$Y$, denoted by $R_i\,{:}\; L^2(\eta_i\vert_Y) \longrightarrow
L^2(\eta_i\vert_Y)$.  This Schr\"odinger operator has the form
\begin{equation}
  R_i = -\triangle_i(Y) + v_i\,,
\end{equation}
where $\triangle_i(Y)$ is the Laplacian with respect to the flat
connection on $L^2(\eta_i\vert_Y)$ and the constant $v_i$ is the value
of $V_i$ on $Y$.  Physically, the spectrum of the operator $R_i$
represents the quantum energy levels of a free particle on the compact
manifold $Y$, with a state space twisteted according to the line
bundle $\eta_i$.  A particle trajectory on $Y$ corresponds to a path
of nuclear configurations within a compact submanifold of the nuclear
configuration space of the molecule.  Vibronic excitations of this
nature are called pseudorotations, because they have the character of
a generalised rotation.  We remark that it is only in the context of
the strong Jahn-Teller coupling approximation that it makes sense to
interpret a special class of vibronic excitations as pseudorotations.
If the Jahn-Teller coupling is not strong, then eigenvectors of the
Born-Oppenheimer Hamiltonian cannot be decomposed into two types of
excitations, pseudorotations and normal oscillations, which are
effectively decoupled.

To find the spectrum of $R_i$, we will need to study its pullback over
the universal covering projections of the leaves in our isoparametric
foliation of $S^4(r_0)$.  The fibre of the universal covering
$p\,{:}\; \widetilde{\mathfrak{F}(b)} \rightarrow \mathfrak{F}(b)$ is
the fundamental group, $\pi_1(\mathfrak{F}(b))$.  For $b$ equal to
either $0$ or $1$, we have that $\pi_1(\mathfrak{R}_1) =
\pi_1(\mathfrak{R}_3) = \pi_1(\mathbb{R}\mathrm{P}(2)) =
\mathbb{Z}_2$.  For $b \in (0,1)$, the fundamental group of
$\mathfrak{F}(b)$ is isomorphic to
$\pi_1(\mathbb{R}\mathrm{F}(1,1,1))$. In the following proposition, we
show that this is isomorphic to $\quat$, the 8-element group
of unit quaternions.
\begin{proposition}
The fundamental group of $\mathbb{R}\mathrm{F}(1,1,1)$ is isomorphic to
$\quat$.
\end{proposition}
\begin{proof}
From Figure~\ref{fig2}, we have that $\mathfrak{W}_2 = \mathfrak{W}_1
\cap \mathfrak{W}_3$, where the homotopy types of these spaces are
$\mathfrak{W}_1 \simeq \mathbb{R}\mathrm{P}(2)$, $\mathfrak{W}_3
\simeq \mathbb{R}\mathrm{P}(2)$, and $\mathfrak{W}_2 \simeq
\mathbb{R}\mathrm{F}(1,1,1)$,  The Mayer-Vietoris homology exact sequence for
$\mathfrak{W}_2 = \mathfrak{W}_1 \cap \mathfrak{W}_3$ gives
$H_1(\mathbb{R}\mathrm{F}(1,1,1);\mathbb{Z}) \cong \mathbb{Z}_2 \oplus \mathbb{Z}_2$.
Furthermore, the homotopy exact sequence for the fibre bundle
\begin{equation}
\parbox{\textwidth}{%
     \xymatrix{%
     \Orth(1) \times \Orth(1) \times \Orth(1)\;\ar @{^{(}->} [r] &
       \Orth(3) \ar[d] \\
     & \mathbb{R}\mathrm{F}(1,1,1)}
}
\end{equation}
implies that $\pi_1(\mathbb{R}\mathrm{F}(1,1,1))$ is an 8-element group.  In
addition, the Hurewicz theorem implies that the
commutator subgroup of $\pi_1(\mathbb{R}\mathrm{F}(1,1,1))$ is $\mathbb{Z}_2$.

Up to isomorphism, there are only two nonabelian groups with 8
elements: the quaternion group $\quat$ and the dihedral
group $\mathcal{D}_8$.  Of these two groups, only $\quat$ can
be the fundamental group of a closed 3-manifold, $\mathcal{M}$.  This is
because the universal covering space of any closed 3-manifold with
finite fundamental group has the homotopy type of $S^3$
\cite[Thm.~3.6]{H76}.  From this it follows that any element of order
two in $\pi_1(\mathcal{M})$ must belong to the centre of $\pi_1(\mathcal{M})$
\cite[Cor.~1]{M57}.  However, $\mathcal{D}_8$ contains an element of
order two which is not in its centre.
\end{proof}

Over $p$ we have the following pullback square,
\begin{equation}
\parbox{\textwidth}{%
  \xymatrix @H=2.7ex @W=3em {%
     & \mathbb{R} \ar @{^(->} [d] & \mathbb{R} \ar @{^(->}[d] \\
     \pi_1(\mathfrak{F}(b))\, \ar @{^(->}[r]  & 
       \widetilde{\mathfrak{X}_i}(b) \ar[r]_{\tilde p} \ar[d]_{\wp_i} & 
       \mathfrak{X}_i(b) \ar[d]^{\wp_i}\\
     \pi_1(\mathfrak{F}(b))\, \ar @{^(->}[r] & \widetilde{\mathfrak{F}}(b) 
       \ar[r]_{p} & \mathfrak{F}(b) \\
  }
}
  \label{pbsquare}
\end{equation}
for each of the real line bundles $\varrho_i(b) =
\eta_i\vert_{\mathfrak{F}(b)}$.  Note that $\varrho_1(b)$ is only
well-defined for $0 \le b < 1,\,\,$ $\varrho_2(b)$ is only well defined
for $0 < b < 1$, and $\varrho_3(b)$ is only well defined for $0 < b
\le 1$. The map $\widetilde \wp_i$ is the projection of the real line
bundle $\widetilde{\varrho_i}$, which is the pullback of $\varrho_i$
over $p$.  The total space of $\widetilde \varrho_i$ is defined as
\begin{equation}
  \widetilde{\mathfrak{X}_i}(b) = \left\{(y,x) \in
  \widetilde{\mathfrak{F}}(b) \times \mathfrak{X}_i(b) \mid p(y) =
  \wp_i(x)\right\}
\end{equation}
and the projection $\widetilde{\wp_i}$ is defined by
$\widetilde{\wp_i}\,{:}\; (y,x) \longmapsto y$.  Likewise, the map
$\widetilde{p}$, which is the pullback over $\wp_i$ of $p$, is defined
by $\widetilde p\,{:}\; (y,x) \longmapsto x$.  Note that the diagram
(\ref{pbsquare}) is a commutative diagram in that $\wp_i \circ
\widetilde{p} = p \circ \widetilde{\wp_i}$.

The universal covering space $\widetilde Y$ has a Riemannian metric
$\widetilde{g}$, which is a lifting of the Riemannian metric
$g$ on $Y$.  Each of the Born-Oppenheimer Hamiltonians
$R_i$,\linebreak[4] for~\mbox{$i = 1,2,3$}, can be lifted to a Hamiltonian
\begin{equation}
  \widetilde R_i = -\widetilde\triangle_i + v_i \,,
\end{equation}
where $\widetilde\triangle_i$ is a Laplacian with respect to the
Riemannian metric $\widetilde{g}$.  This operator acts on
$L^2(\widetilde \eta_i)$, the Hilbert space of $L^2$ sections of
$\widetilde \eta_i$.  Note that an $L^2$ section of $\widetilde
\eta_i$ is a square integrable function $\sigma\,{:}\; \widetilde Y
\longrightarrow \widetilde X_i$ such that the composition $\widetilde
\wp_i \circ \sigma$ is the identity almost everywhere on $\widetilde
Y$.  This implies that $\sigma$ must be of the form
\begin{equation}
  \sigma\,{:}\; w \longmapsto \left( w \,,\, \widehat\sigma (w)\right)\, ,
\end{equation}
where $\widehat\sigma = {\widetilde p} \circ \sigma$ is a square
integrable function from $\widetilde Y$ to $X_i$, satisfying $\wp_i
\circ \widehat\sigma = p$.

Now consider a section $\phi \in L^2 (\eta_i)$. This is a
square integrable function $\phi\,{:}\; Y \rightarrow X_i$ such that
$\wp_i \circ \phi$ is the identity almost everywhere on $N_2$.  From
$\phi$, we construct a function $\widehat\phi\,{:}\; \widetilde Y
\rightarrow X_i$ by defining $\widehat\phi = \phi \,\circ\, p$.  The
pullback of $\phi$ is then $\widetilde\phi\,{:}\; \widetilde Y
\rightarrow \widetilde X_i$, defined by $\widetilde \phi (w) = (w\, ,
\, \widehat\phi (w))$.  This construction of pulling back a section
from $\eta_i$ to $\widetilde\eta_i$ constitutes a one-to-one linear
map $\Upsilon$ of $L^2 (\eta_i)$ into $L^2 (\widetilde\eta_i)$,
because $p$ is a finite-to-one covering projection.

Each element $g \in \pi_1 (Y)$ defines a deck
transformation\footnote{A deck transformation $D$ of the universal
covering space projection $p$ is a homeomorphism $D \,{:}\;
\widetilde Y \rightarrow \widetilde Y$ satisfying $p \circ D =
p$.  For more details on deck transformations, see Section~III.6
of~\cite{B93}.} $D(g)\,{:}\; \widetilde Y \rightarrow
\widetilde Y$, because the group $\mathcal{D}$ of all deck
transformations on $\widetilde Y$ is canonically isomorphic to
$\pi_1 (Y)$.  Each deck transformation $D_g \in \mathcal{D}$ induces a
linear operator $\Lambda(g)$ on $L^2 (\widetilde \eta_i)$, defined by
mapping the section
\begin{displaymath}
  \sigma\,{:}\; w \longmapsto \left(w \,,\, \widehat\sigma (w) \right)
\end{displaymath}
to the section
\begin{equation}
  \Lambda(g)\, \sigma \,{:}\;\, w \longmapsto 
           \left( w \,,\, \widehat\sigma (D^{-1}_g(w)) \right) .
\end{equation}
To compute the spectrum of $R_i$, we will use the following lemma.
\begin{lemma}
\label{cover}
Suppose that $p\,{:}\, \widetilde M \rightarrow M$ is a cover map and
that $\eta$ is a vector bundle over $M$.  Let $\widetilde \eta$ denote
the pullback of $\eta$ over $p$.  If $\phi \in L^2(\eta)$ is an
eigenvector for a differential operator $R \,{:}\, L^2(\eta)
\rightarrow L^2(\eta)$ with eiganvalue $\lambda$, then the pullback
section $\widetilde \phi = \Upsilon(\phi)$ is an eigenvector of the
pullback operator $\widetilde R$ with eigenvalue $\lambda$.
Conversely, if $\widetilde \phi \in L^2(\widetilde \eta)$ is an
eigenvector of $\widetilde R$ with eigenvalue $\lambda$ and
$\widetilde \phi$ is in the image of $\Upsilon$, then $\widetilde
\phi$ is the pullback of an eigenvector of $R$, with eigenvalue
$\lambda$. 
\end{lemma}
We remark that this result for the special case of Laplacians acting
on functions is given in~\mbox{\cite[Section~III.A.II]{BGM74}}
and~\mbox{\cite[pp.~27--28]{C84}}.
\begin{proof}
Denote the projection map and total space of $\eta$ by $q$ and $E$,
respectively.  Then, the total space of $\widetilde\eta$ is
\begin{equation}
\left\{(x,y) \in \widetilde M \times E \mid p(x) = q(y)\right\}
\end{equation}
and the projection map is $\widetilde q \,{:}\, (x,y) \rightarrow x$.

Suppose that $\phi \in L^2(\eta)$ satisfies $R\phi = \lambda\phi$.
Then the pullback section is $\widetilde\phi(x) = (x,\phi\circ
p(x))$.  It follows from the definition of $\widetilde R$ that
\begin{equation}
  \widetilde R \, \widetilde \phi = (x \,,\, R\,\phi \circ p) 
  = (x \,,\, \lambda \phi \circ p) = \lambda \widetilde\phi\,.
\end{equation}

Now consider $\widetilde\phi \in \mathop{\rm Image}(\Upsilon)$
satisfies $\widetilde R \, \widetilde\phi = \lambda\widetilde\phi$.
There exists a unique $\phi \in L^2(\eta)$ such that $\Upsilon(\phi) =
\widetilde\phi$.  From the construction of $\widetilde R$, we have
\begin{equation}
  \widetilde R \, \widetilde\phi = \left(x\,,\,R\,\phi(p(x))\right)
  \quad\text{and}\quad
  \widetilde R \, \widetilde\phi = \left( x\,,\,
  \lambda\phi(p(x))\right)\,.
\end{equation}
Therefore, $R\,\phi = \lambda\phi$, because $p$ is a covering. 
\end{proof}

Note that the image of the map $\Upsilon \,{:}\; L^2(\eta_i)
\rightarrow L^2(\widetilde\eta_i)$ is the fixed point set of the group
action
\begin{equation}
   \Lambda :\,  \pi_1 (Y)  \times  L^2(\widetilde\eta_i)
   \longrightarrow L^2(\widetilde\eta_i) \,.
   \label{action}
\end{equation}
Therefore, we may use the action $\Lambda$ of the group $\pi_1(Y)$ on
$L^2(\widetilde\eta_i)$ and Lemma~\ref{cover} to obtain the spectrum
of $B_i$ from the spectrum of $\widetilde B_i$.

The real line bundle $\widetilde\eta_1$ is isomorphic to the trivial
line bundle, because $\widetilde Y$ is simply connected.  Therefore,
it has a smooth globally-defined normalised section.  In fact, there
exist exactly two normalised sections and we chose one to denote by
$\tau$.  Of course, the other section will then be $-\tau$ and the set
of normalised sections is $\Gamma(\widetilde\eta_i) =
\left\{\tau,-\tau\right\}$.  Nontriviality of $\eta_i$ implies that
the action of $\Lambda$ on $\Gamma(\widetilde\eta_i)$ must be
transitive. Otherwise, $\tau$ would be invariant under $\Lambda$ and
could be pushed down to give a normalised section of $\eta_i$, which
is contrary $\eta_i$ being nontrivial. Thus, the action $\Lambda$ is
\begin{equation}
  \Lambda_\mathbf{1}\,\tau = \tau \qquad
  \Lambda_\mathbf{-1}\,\tau = -\tau\,.
  \label{normtrans}
\end{equation}

Using the section $\tau$, we are able to represent $L^2$ sections of
$\widetilde\eta_i$ as elements of $L^2(\widetilde Y;\,\mathbb{R})$, the
real Hilbert space of real-valued square-integrable functions on
$\widetilde Y$.  A section $\sigma\,{:}\; \widetilde Y \rightarrow
\widetilde X_1$ is represented by the function
\begin{equation}
  f \,{:}\; \widetilde Y \longrightarrow \mathbb{R} \,,\quad 
     \mbox{defined by}\quad f(w) = \langle \tau (w) \, , \, 
     \sigma (w)\rangle \,, \label{def1}
\end{equation}
where $\langle \; \cdot \;,\; \cdot \; \rangle$ is the fibre-wise
inner product on $\widetilde\eta_i$.  This correspondence is a Hilbert
space isomorphism from $L^2(\widetilde\eta_1)$ to $L^2(\widetilde
Y;\,\mathbb{R})$.  It follows from~(\ref{normtrans}) and~(\ref{def1})
that $\sigma$ is a $\Lambda$-invariant section of $\widetilde\eta_i$
if and only if the function $f$ corresponding to $\sigma$ satisfies
\begin{equation}
  \Lambda_\mathbf{1}\,f = f \qquad \Lambda_\mathbf{-1} = -f\,.
\end{equation}

Now consider the case when $b_0 \in (0,1)$.  Before computing the
action $\Lambda$ of $\pi_1(Y) \cong \quat$ on each of the
Hilbert spaces $L^2(\widetilde \eta_i)$, some basic facts about
$\quat$ will be reviewed.  The nonabelian group
$\quat$ has an order two commutator subgroup.  We label the
unit in $\quat$ by \mbox{\boldmath $1$} and the nontrivial
element in the commutator subgroup by \mbox{\boldmath $-1$}.
Following convention, the other six elements of $\quat$ are
denoted by \mbox{\boldmath $\pm i$}, \mbox{\boldmath $\pm j$}, and
\mbox{\boldmath $\pm k$}.  The nontrivial proper subgroups of
$\quat$ are the commutator subgroup $\{\mbox{\boldmath $\pm
1$}\}$ and the three abelian subgroups of order four,
$\{\mbox{\boldmath $\pm 1$}\,,\,\mbox{\boldmath $\pm i$}\}$,
$\{\mbox{\boldmath $\pm 1$}\,,\,\mbox{\boldmath $\pm j$}\}$, and
$\{\mbox{\boldmath $\pm 1$}\,,\,\mbox{\boldmath $\pm k$}\}$.  Each of
the latter three subgroups is isomorphic to the cyclic group
$\mathbb{Z}_4$, because each contains an element of order four.

Note that each of the three real line bundles $\widetilde\eta_i$ is
isomorphic to the trivial real line bundle over $\widetilde Y$,
because $\widetilde Y$ is simply connected.  Therefore, there exists a
smooth globally-defined normalised section for $\widetilde\eta_i$.  As
before, there are exactly two such sections, related by multiplication
by $-1$ and the group $\pi_1(Y)$ acts on $\Gamma(\widetilde \eta_i)$,
the set of normalised sections.  Nontriviality of $\eta_i$
implies that $\Lambda$ must be a transitive group action of $\pi_1(Y)$
on the set $\Gamma(\widetilde \eta_i)$.  Therefore, the isotropy
subgroup of this action must be one of the three order four subgroups
of $\quat$.  Each of these isotropy subgroups corresponds to
one of the three line bundles $\eta_i$, for $i = 1,2,3$

The precise correspondence between the line bundle $\eta_i$ and the
order four subgroups of $\pi_1(Y)$ can be determined from the the
classification theorem for covering
projections~\cite[Section~2.5]{S66}.  This theorem states that there is
a one-to-one correspondence between covering projections from a
connected covering space to $Y$ and subgroups of $\pi_1(Y)$, where two
subgroups are considered as equivalent if they are conjugate to each
other.  Therefore, $Y$ has exactly three 2\nolinebreak\mbox{-}fold
covering projections, corresponding to the three inequivalent order
four subgroups of $\pi_1(Y)$.  Now consider that associated with the
real line bundles $\eta_i$ are three principal ${\bf O}(1)$ bundles,
$P(\eta_i)$, for $i = 1,2,3$.  These principal bundles are
inequivalent 2\nolinebreak\mbox{-}fold coverings of $Y$, and
therefore, they must correspond to the three order four subgroups of
$\pi_1(Y)$.  This correspondence identifies the fundamental group of
the total space $P(X_i)$ of $P(\eta_i)$ with one of the order four
subgroups of $\pi_1(Y)$.

It remains to verify that this identification of the line bundles
$\eta_i$ with the order four subgroups of $\pi_1(Y)$ is reflected in
the isotropy subgroup of the action $\Lambda$ on $\Gamma(\widetilde
\eta_i)$.  We denote the pullback of each of the line bundles $\eta_i$
over the projection $P(\wp_i)$ of its associated principal bundle
$P(\eta_i)$ by $\eta^\prime_i$.  As remarked above, each
$\eta^\prime_i$ is isomorphic to the trivial real line bundle over
$P(X_i)$.  In the same vein as the action $\Lambda$ constructed above,
there is an action of $\Orth(1)$ on $\Gamma(\eta^\prime_i)$, the two
element set of normalised sections in $\eta^\prime_i$.  Furthermore,
this action must be transitive because $\eta_i$ is a nontrivial line
bundle.  From each of the 2\nolinebreak\mbox{-}fold covering spaces
$P(X_i)$, we obtain the universal covering space $Y$ by taking a
further 4\nolinebreak\mbox{-}fold covering with projection $c_i$.  The
fibre of this 4\nolinebreak\mbox{-}fold covering is the fundamental
group of $P(X_i)$.  The three distinct ways of obtaining the universal
covering projection $p\,{:}\;\widetilde Y \longrightarrow Y$ as
the composition of a 4\nolinebreak\mbox{-}fold covering projection and
a 2\nolinebreak\mbox{-}fold covering projection are shown in the
diagram below.
\begin{equation}
\parbox{\textwidth}{%
  \xymatrix{%
     & \widetilde Y \ar[dl]_(0.6){c_1} \ar[d]^(0.47){c_2} 
        \ar[dr]^(0.6){c_3} & \\
     P(X_1) \ar[dr]_(0.33){P(\wp_1)} & P(X_2) \ar[d]^{P(\wp_2)}
       & P(X_3) \ar[dl]^(0.33){P(\wp_3)} \\
     & Y &
     }
}\label{3cover}
\end{equation}
Note that the pullback of $\eta^\prime_i$ over $c_i$ is the pullback
of a trivial line bundle.  Therefore, sections in $\Gamma(\widetilde
\eta_i)$ are invariant under the action $\Lambda$ induced by the deck
transformations of this covering projection.  Thus, the deck
transformations corresponding to the order four subgroup of
$\pi_1(Y)$ which is associated with $\eta_i$ is the isotropy
subgroup of the action $\Lambda$ on $\Gamma(\widetilde \eta_i)$.

It is simply a matter of labelling convention how we chose to
designate the isotropy subgroups of $\Lambda$ as order four subgroups
in $\quat$.  We shall denote the isotropy subgroup of $\Lambda$
on $\Gamma(\widetilde \eta_1)$ by $\{\mbox{\boldmath $\pm
1$}\,,\,\mbox{\boldmath $\pm i$}\}$, the isotropy subgroup of
$\Lambda$ on $\Gamma(\widetilde \eta_2)$ by $\{\mbox{\boldmath $\pm
1$}\,,\,\mbox{\boldmath $\pm j$}\}$ and the isotropy subgroup of
$\Lambda$ on $\Gamma(\widetilde \eta_3)$, and $\{\mbox{\boldmath $\pm
1$}\,,\,\mbox{\boldmath $\pm k$}\}$.  Therefore, any normalised
section $\tau$ in $\Gamma (\widetilde \eta_1)$ tranforms according to
\begin{equation}
  \begin{alignedat}{2}
     \Lambda_{\pm\mathbf{1}}\,\tau &= \tau 
       \qquad & \qquad \Lambda_{\pm\mathbf{i}} \,\tau &= \tau \\
     \Lambda_{\pm\mathbf{j}}\,\tau &= -\tau 
       \qquad & \qquad \Lambda_{\pm\mathbf{k}}\,\tau &= -\tau
  \end{alignedat} \label{deck_trans}
\end{equation}
There are similar transformation equations for sections in $\Gamma
(\widetilde \eta_2)$ and $\Gamma (\widetilde \eta_3)$.

Given a choice of normalised section $\tau \in \Gamma(\widetilde
\eta_i)$, $L^2$ sections of $\widetilde \eta_i$ are represented as
functions by equation~(\ref{def1}).

It follows from (\ref{deck_trans}) that $\sigma$ is a
section in $\widetilde \eta_1$ satifying $\Lambda_g \sigma = \sigma$
for all $g \in \pi_1(N_2)$ if and only if the associated function $f$
satisfies
\begin{subequations}
\begin{gather}
  \begin{alignedat}{2}
     \Lambda_{\pm\mathbf{1}}\, f &= f 
        \qquad & \qquad \Lambda_{\pm\mathbf{i}} \, f &= f \\
     \Lambda_{\pm\mathbf{j}}\, f & = - f 
        \qquad & \qquad \Lambda_{\pm\mathbf{k}} \, f &= - f
  \label{transprop-a}
  \end{alignedat}
\intertext{Similary, a section $\sigma$ in $\widetilde \eta_2$ or
$\widetilde \eta_3$ satisfies $\Lambda_g\sigma = \sigma$ for all $g
\in \pi_1(N_2)$ if and only if the associated function satisfies}
  \begin{alignedat}{2}
     \Lambda_{\pm\mathbf{1}}\, f &= f 
        \qquad & \qquad \Lambda_{\pm\mathbf{i}} \, f &= -f \\
     \Lambda_{\pm\mathbf{j}}\, f & = f 
        \qquad & \qquad \Lambda_{\pm\mathbf{k}} \, f &= - f
  \label{transprop-b}
  \end{alignedat}
\intertext{and}
  \begin{alignedat}{2}
     \Lambda_{\pm\mathbf{1}}\, f &= f 
        \qquad & \qquad \Lambda_{\pm\mathbf{i}} \, f &= -f \\
     \Lambda_{\pm\mathbf{j}}\, f & = - f 
        \qquad & \qquad \Lambda_{\pm\mathbf{k}} \, f &= f
  \label{transprop-c}
  \end{alignedat}
\end{gather}
\label{transprop}
\end{subequations}
respectively.

Under the Hilbert space isomorphism (\ref{def1}), the Schr\"odinger
operators $\widetilde B_i$, for $i = 1,2,3$, correspond to the
operators
\begin{equation}
  \widehat B_i = -\widehat\triangle_2 + \widehat V_i \,,
\end{equation}
acting on the function space $L^2(\widetilde N_2;\,\mathbb{R})$.  The
Laplacian $\widehat\triangle_2$ is defined with respect to the
Riemannian metric $\widetilde{g}$.  Viewed as a function on
$N_2$, the potential function $\widehat V_i$ is the same as
$\widetilde V_i$; however, they should be distinguished as operators
because $\widetilde V_i$ acts on $L^2(\widetilde \eta_i)$, whereas
$\widehat V_i$ acts on $L^2(N_2;\,\mathbb{R})$.  The spectrum of the
Born-Oppenheimer Hamiltonian $B_i$ will be obtained by finding all
eigenfunctions of $\widehat B_i$ which have the transformation
properties in (\ref{transprop}) for $H_1$, or corresponding
transformation properties for $H_2$ and $H_3$.

Recall that the potential $V_i$ is minimised on the submanifold $Y =
\mathfrak{F}(b_0)$.  Therefore, the potential $\widehat V_i$
is minimised on the covering space $\widetilde Y = p^{-1}(Y)$.  The
submanifold $\widetilde Y \subset \widetilde N_2$ is diffeomorphic to
the 3\nolinebreak\mbox{-}dimensional sphere $S^3$.  The Riemannian
metric on $\widetilde Y$ is $\widetilde{g}$, which has the
same local curvature as $Y$ with Riemannian metric $g$,
because $\widetilde Y$ is a covering space of $Y$.

\section{Spectrum on the Projective Spaces}

In this section we calculate the spectrum of the Laplacian with
respect to the flat connection on the line bundles $\eta_1$ and
$\eta_3$, restricted to $\mathfrak{R}_1$ and $\mathfrak{R}_3$,
respectively.  The fundamental group $\pi_1(Y)$ is now $\mathbb{Z}_2$
and the pullback square~(\ref{pbsquare}) becomes
\begin{equation}
\parbox{\textwidth}{%
  \xymatrix @H=2.7ex @W=3em {%
     & \mathbb{R} \ar @{^(->} [d] & \mathbb{R} \ar @{^(->}[d] \\
     \mathbb{Z}_2 \ar @{^(->}[r]  & 
       \widetilde X_i \ar[r]_{\tilde p_i} \ar[d]_{\wp_i} & 
       X_i \ar[d]^{\wp_i}\\
     \mathbb{Z}_2 \ar @{^(->}[r] & \widetilde{\mathfrak{R}_i} 
       \ar[r]_{p_i} & \mathfrak{R}_i \\
  }
}
\label{square1}
\end{equation}
where $i=1,3$.  The line bundles $\wp_i$ are isomorphic to the
canonical line bundle over $\mathbb{R}{\rm P}(2)$.  Its pullback
$\widetilde\eta_1$ is easily seen to be isomorphic to the trivial line
bundle, because $\widetilde N_1$ is simply connected.  However, it
will be useful to note that the bundle $p_1$ is isomorphic to the
associated principal bundle for the real line bundle $\wp_1$.  It is
generally true that the pullback of a vector bundle over its
associated principal bundle is isomorphic to the trivial bundle.

We will refer to the trivial element in $\pi_1(Y)=\mathbb{Z}_2$ by
$\mathbf{0}$ and the nontrivial element by~$\mathbf{1}$.  The fact
that the action $\Lambda$ defined in (\ref{action}) is a group action
implies that $\Lambda_\mathbf{0}\,\tau = \tau$, and that
$\Lambda_\mathbf{1}\,\tau$ is either $\tau$ or $-\tau$.  Observe that
if $\Lambda_\mathbf{1}\,\tau$ were equal to $\tau$, then $\tau$ would
be a normalised section in $\widetilde\eta_1$, which could be obtained
as the pullback of a normalised section in $\eta_1$.  However, this is
impossible because $\eta_1$ is a nontrivial real line bundle.  In
summary, nontriviality of $\eta_1$ implies that $\Lambda$ must be a
transitive group action of $\mathbb{Z}_2$ on the set $\{\tau,
-\tau\}$.  Therefore,
\begin{equation}
  \Lambda_\mathbf{0}\,\tau = \tau \qquad\mbox{and}\qquad
  \Lambda_\mathbf{1}\,\tau = -\tau \, . 
  \label{ns} 
\end{equation}
It follows from (\ref{def1}) and (\ref{ns}) that $\psi$ is a section
of $\widetilde\eta_1$ satisfying $\Lambda_g\,\psi =
\psi$ for all $g \in \mathbb{Z}_2$, if and only if
\begin{equation}
  \Lambda_\mathbf{0} \, f = f 
  \qquad\mbox{and}
  \qquad \Lambda_\mathbf{1}\, f = -f \, . 
  \label{lemmaf}
\end{equation}
where the function $f$ is associated to $\psi$ by (\ref{def1}) and the
action $\Lambda$ on functions is defined by $\Lambda\,f(w) =
f(D^{-1}_g(w))$.

\begin{theorem}
\label{thmR1}
The eigenvalues of $-\triangle_1(\mathfrak{R}_1)$ are
\begin{subequations}
\begin{gather}
  \lambda_n = \frac{n(n+1)}{3r_0^2}\,,\quad\text{for $n=1,3,5,\dots$}\,.
\intertext{The multiplicities of these eigenvalues are}
  \mathop{\rm mult}(\lambda_n) = 2n+1\,.
\end{gather}
\end{subequations}
The spectrum of $-\triangle_3(\mathfrak{R}_3)$ is the same as the spectrum of
$-\triangle_1(\mathfrak{R}_1)$.
\end{theorem}
\begin{proof}
Considering first $\eta_1|_{\mathfrak{R}_1}$, the associated principal
bundle is the nontrivial $\mathbb{Z}_2$-bundle
\begin{equation}
\parbox{\textwidth}{%
   \xymatrix{%
    \mathbb{Z}_2\> \ar @{^{(}->} [r] &
       S^2(\sqrt{3}\>r_0) \ar[d] \\
     & \mathfrak{R}_1}
}
    \label{princR1}
\end{equation}
The pullback of $\eta_1|_{\mathfrak{R}_1}$ over its associated
principal bundles is $\widetilde \eta_1|_{\mathfrak{R}_1}$, which is
isomorphic to the trivial line bundle of $\mathfrak{R}_1$.  It follows
from Lemma~\ref{cover} that the eigenvalues of $\triangle_1$ correspond to
the eigenfunctions of the pullback Laplacian $\widetilde \triangle_1$
on $S^2(\sqrt{3}r_0)$, which transform according to~(\ref{lemmaf}).

The Laplacian $\widetilde \triangle_1$ is the standard Laplacian on
$S^2(\sqrt{3}r_0)$, because $\triangle_1$ is the Laplacian with
respect to the flat connection on $\eta_1|_{\mathfrak{R}_1}$.
Therefore, the eigenvectors of $-\widetilde \triangle_1$ are the
restrictions to $S^2(\sqrt{3}r_0)$ of the harmonic polynomials in
$\mathbb{R}^3$~\cite{VK93}.  These polynomials constitute the kernel of the usual
Laplacian on $\mathbb{R}^3$.  The eigenvalue associated to a harmonic
polynomial of degree $n$ is
\begin{subequations}
\begin{gather}
  \lambda_n = \frac{n(n+1)}{3r_0^2}\,,\quad\text{for
  $n=0,1,2,\dots$.} 
\intertext{The multiplicity of $\lambda_n$ is}
  \mult(\lambda_n) = 2n+1
\end{gather}
\end{subequations}
The transformation properties in~(\ref{lemmaf}) imply that the
eigenvalues of $\triangle_1$ correspond to the harmonic polynomials of
odd degree.

The spectrum of the Laplacian on $\eta_3|_{\mathfrak{R}_3}$ is the same
as the spectrum on $\eta_1|_{\mathfrak{R}_1}$, because the two bundles
are isometric.
\end{proof}

\section{Spectrum with Constant Curvature}

Although the universal covering space of $Y = \mathfrak{F}(b)$, for $b
\in (0,1,)$ is diffeomorphic to $S^3$, it is not isometric to a
standard sphere $S^3(r)$ with constant curvature $\frac{1}{r^2}$.
This may be verified by calculating $K_i$, the sectional curvature of
$\mathfrak{F}(b)$ in planes orthogonal to the principal foliation
$\mathcal{C}_i(b)$.  These sectional curvatures are
\begin{equation}
  K_1^S = \frac{2\,(1 - b + b^2)}{3\,(1 - b) \, r_0^{\> 2}}\,, \quad
  K_2^S = \frac{2\,(1 - b + b^2)}{3\,b\,(1-b)\, r_0^{\> 2}}\,, 
  \;\;\text{and}\;\;\,
  K_3^S = \frac{2\,(1 - b + b^2)}{3\,b\, r_0^{\> 2}}\,.
\end{equation}
The universal covering space would have constant curvature if and only
if $Y$ had constant curvature.

However, because it is covered by a sphere, $Y$ does admit a constant
curvature Riemannian structure~\cite[Thm.~5.1.2]{W74}.  The notation
$Y_K$ will be used to denote the manifold $Y$ with curvature $K$.  The
universal covering space, $\widetilde Y_K$, of $Y_K$ is isometric to
the 3-sphere of radius $r = \frac{1}{\sqrt{K}}$.  It is instructive to
first carry out the computation the spectra of the bundle Laplacians
$\triangle_i(Y_K)$, for $i=1,2,3$, on $Y_K$, before confronting this
problem for $Y$ with the more complicated isoparametric Riemannian
structure.

The pullback of $\triangle_i(Y_K)$ with respect to the flat connection
on $\eta_i$, is $\triangle(S^3(r))$, the standard Laplacian on the
3-sphere.  All of the eigenvalues of $\triangle(S^3(r))$ are obtained
as the restrictions to $S^3(r)$ of the harmonic polynomials in
$\mathbb{R}^4$.  We denote the vector space of harmonic polynomials of
degree $n$ on $\mathbb{R}^4$ by $\mathcal{H}_n$.  This is the eigenspace
of $-\triangle(S^3(r))$ with eigenvalue
$\lambda_n=\frac{n(n+2)}{r^2}$~\cite{VK93}.  To determine the
dimensions of the eigenspaces of $\triangle(S^3(r))$, we first consider
the vector space $\mathcal{P}(n)$, consisting of all degree $n$ homogeneous
polynomials on $\mathbb{R}^4$.  The dimension of $\mathcal{P}(n)$ is
\begin{equation}
  \dim\left(\mathcal{P}(n)\right) = \binom{n + 3}{n}\,.
\end{equation}
The Laplacian on $\mathbb{R}^4$ is an onto linear map
\begin{equation}
  \triangle(\mathbb{R}^4)\,{:}\;\, \mathcal{P}(n) \longtwoheadrightarrow \mathcal{P}(n - 2)\,,
  \label{lap4}
\end{equation}
where $\mathcal{P}(n)$ is interpreted as the
0\nolinebreak\mbox{-}dimensional vector space when $n < 0$.
Recalling that $\mathcal{H}(n)$ is the kernel of the map in~(\ref{lap4}), we
obtain that the multiplicity of $\lambda_n$ is
\begin{equation}
\begin{aligned}
  \dim\left(\mathcal{H}(n)\right) &= 
  \dim\left(\mathcal{P}(n)\right) - \dim\left(\mathcal{P}(n-2)\right)\\
  & = 
  \left(n + 1\right)^2\,, \quad \text{for $n = 0,1,2,\dots$.}
\end{aligned}
\end{equation}
Using Lemma~\ref{cover}, we obtain the eigensections of
$-\triangle_i(Y_K)$, by finding all harmonic polynomials in
$\mathbb{R}^4$ which transform accordingly under the action of
$\Lambda$ in (\ref{transprop}).

To explicitly construct the action $\Lambda$, use the fact that a
vector $(x_0,x_1,x_2,x_3) \in \mathbb{R}^4$ can be
represented as a quaternionic number $x \in \mathbb{H}$ by setting
\begin{equation}
  x = x_0 + x_1{\bf i} + x_2{\bf j} + x_3{\bf k}
\end{equation}
The group $\quat$ acts on the sphere $S^3(r) \subset
\mathbb{H}$ by left multiplication\footnote{Our choice of left
multiplication rather than right multiplication is simply a matter of
convention.}.  This action of $\quat$ on $S^3(r)$ induces the
action $\Lambda$ that is defined in Section~\ref{Sect_CS}.
Specifically,
\begin{subequations}
\begin{align}
  \Lambda_\mathbf{\pm 1} & :\;  f(x_0,x_1,x_2,x_3) 
     \longmapsto f(\pm x_0, \pm x_1 , \pm x_2 , \pm x_3) \\
  \Lambda_\mathbf{\pm i} & :\;  f(x_0,x_1,x_2,x_3)
     \longmapsto f(\pm x_1, \mp x_0 , \pm x_3 , \mp x_2) \\
  \Lambda_\mathbf{\pm j} & :\;  f(x_0,x_1,x_2,x_3) 
     \longmapsto f(\pm x_2 , \mp x_3 , \mp x_0 , \pm x_1) \\
  \Lambda_\mathbf{\pm k} & :\;  f(x_0,x_1,x_2,x_3) 
     \longmapsto f(\pm x_3 , \pm x_2 , \mp x_1 , \mp x_0)
\end{align}
\end{subequations}
Comparing this action to the transformation properties
in~(\ref{transprop}) allows us to prove the following theorem.

\begin{theorem}
The three Laplacians $\triangle_i(Y_K)$, $i=1,2,3$, have the same
spectrum.  The eigenvalues of $-\triangle_i(Y_K)$ are
\begin{subequations}
\begin{gather}
  \lambda_n = r^{-2}(n+1)(n+2)\,,\qquad\text{for $\,n=0,1,2,\dots$}
\intertext{The multiplicity of $\lambda_n$ is}
  \mult(\lambda_n) = (n\backslash 2\,+1)(2n+3)\,,
\end{gather}
\end{subequations}
where $\backslash\,$ is the integer division operator.
\end{theorem}

\begin{proof}
In order to find a basis of harmonic polynomials on $\mathbb{R}^4$
which have the required transformation properties under $\Lambda$, we
first note that the action $\Lambda$ commutes with the Laplacians.
Therefore, we can proceed by first finding the subspace of homogeneous
polynomials which have the transformation properties given
by~(\ref{transprop}) and then find the harmonic polynomials within
this subspace.  

Begin by considering the subspace $\mathcal{P}_1(n)$, which is defined as
the subspace of all degree $n$ polynomials that transform according to
(\ref{transprop-a}).  The invariance of $f \in \mathcal{P}_1(l)$ under
$\Lambda_{\bf -1}$ implies that the degree of $f$ must be even.  To
consider the remaining elements of $\quat$, it will suffice to
check only $\Lambda_\mathbf{i}$ and $\Lambda_\mathbf{j}$, because
$\mathbf{k} = \mathbf{i}\, \mathbf{j}\,$.

Our analysis will be organised according to the number of distinct
exponents in the homogeneous polynomials $f(x_0,x_1,x_2,x_3) =
x_0^{n_0}\,x_1^{n_1}\,x_2^{n_2}\,x_3^{n_3}$, where $l = n_0 + n_1 +
n_2 + n_3$ is even.  First, consider polynomials for which the
exponent of $x_i$ is the same for each $i = 1,2,3$.  A polynomial of
the form $f(x_0,x_1,x_2,x_3) = x_0^n\,x_1^b\,x_2^n\,x_3^n$ is
invariant under both $\Lambda_\mathbf{j}$ and $\Lambda_\mathbf{k}$,
which is inconsistent with (\ref{transprop}).  Therefore, there are no
nontrivial polynomials of this form in $\mathcal{P}_1(l)$.

Next we consider homogeneous polynomials for which two of the
variables have one exponent $n_1 \ge 0$ and the other two variables
have a distinct exponent $n_2 \ge 0$.  The transformation of such
polynomials in $\mathcal{P}^\prime_1(l)$ is given in tabular form below.
\begin{equation}
\setlength{\arraycolsep}{1em}
\renewcommand{\arraystretch}{1.1}
\begin{array}{@{}l|ccc@{}}
          & \Lambda_{\bf i} & \Lambda_{\bf j} & \Lambda_{\bf k} \\ 
\hline
f_1=(x_0 x_1)^{n_1}(x_2 x_3)^{n_2}
          & (-1)^{n_1 + n_2}f_1 & (-1)^{n_1 + n_2}f_6 & f_6 \\
f_2=(x_0 x_2)^{n_1}(x_1 x_3)^{n_2}
          & f_5 & (-1)^{n_1 + n_2}f_2 & (-1)^{n_1 + n_2}f_5 \\
f_3=(x_0 x_3)^{n_1}(x_1 x_2)^{n_2}
          & (-1)^{n_1 + n_2} f_4 & f_4 & (-1)^{n_1 + n_2} f_3 \\
f_4=(x_1 x_2)^{n_1}(x_0 x_3)^{n_2}
          & (-1)^{n_1 + n_2} f_3 & f_3 & (-1)^{n_1 + n_2} f_4 \\
f_5=(x_1 x_3)^{n_1}(x_0 x_2)^{n_2}
          & f_2 & (-1)^{n_1 + n_2} f_5 & (-1)^{n_1 + n_2} f_2 \\
f_6=(x_2 x_3)^{n_1}(x_0 x_1)^{n_2}
          & (-1)^{n_1 + n_2} f_6  &  (-1)^{n_1 + n_2} f_1 & f_1
\end{array} \label{transa}
\end{equation}
The subspace $\mathcal{P}^\prime(l) \subset \mathcal{P}(l)$, defined as the span
of all polynomials of this form for any values of $n_1$ and $n_2$ has
dimension
\begin{equation}
  \dim\left(\mathcal{P}^\prime\left(l\right)\right)
  = \begin{cases}
    \frac{3l}{2}\,, & \text{if $l \equiv 0\pmod{4}$} \\
    \frac{3(l + 2)}{2}\,,& \text{if $l \equiv 2\pmod{4}$}
    \end{cases}
\end{equation}
It follows from (\ref{transa}) that if $n_1 + n_2$ is odd, then a
basis of polynomials with the transformation
properties~(\ref{transprop}) is $\left\{f_2 + f_5 \,,\, f_3 -
f_4\right\}$.  If $n_1 + n_2$ is even, then a basis for the
polynomials which transform as required is $\left\{f_1 - f_6\right\}$.
Therefore, the dimension of $\mathcal{P}^\prime(l) \cap \mathcal{P}_1(l)$ is
\begin{equation}
  d_1(l) = \begin{cases}
           \frac{l}{4}\,, & \text{for $l = 0,2,4,\dots$}\\
           \frac{l + 2}{2}\,, & \text{for $l = 1,3,5,\dots$}
           \end{cases}
\label{dim1}
\end{equation}

Finally, consider $\mathcal{P}^{\prime\prime}(l)$, the vector space complement
of both of the two subspaces of $\mathcal{P}(l)$ already considered.  The
dimension of $\mathcal{P}^{\prime\prime}(l)$ is
\begin{equation}
  \dim\left(\mathcal{P}^{\prime\prime}\left(l\right)\right)
  = \begin{cases}
    \binom{l + 3}{3} - \frac{3l + 2}{2}\,,& \text{if 
              $l \equiv 0\pmod{4}$} \\
    \binom{l + 3}{3} - \frac{3l + 6}{2}\,,& \text{if
              $l \equiv 2\pmod{4}$}
    \end{cases}
\end{equation}
The transformations $\Lambda_\mathbf{i}$, $\Lambda_\mathbf{j}$, and
$\Lambda_\mathbf{k}$ all act on $\mathcal{P}^{\prime\prime}(l)$ without
nontrivial fixed points and the subspace of polynomials that transform
according to (\ref{transprop-a}) in $\mathcal{P}^{\prime\prime}(l)$ can be
constructed by noting that for any $f \in \mathcal{P}^{\prime\prime}(l)$,
the polynomial $f + \Lambda_\mathbf{i}f - \Lambda_\mathbf{j}f -
\Lambda_\mathbf{k}f$ transforms as required.  Therefore, the dimension
of $\mathcal{P}^{\prime\prime}(l) \cap \mathcal{P}_1(l)$ is
\begin{equation}
  d_2(l) = \frac{1}{4}\dim\left(\mathcal{P}^{\prime\prime}
  \left(l\right)\right)
  = \begin{cases}
     0\,, & \text{for $m = 0$}\\[0.3ex]
     2\binom{m + 2}{3} - \frac{m+1}{2}\,, & \text{for $m=1,3,5\dots$}\\[0.3ex]
     2\binom{m + 2}{3} - \frac{m}{2}\,, & \text{for $m = 2,4,6,\dots$}
    \end{cases}
\label{dim2}
\end{equation}
where $l = 2m$ and $m = 0,1,2,\dots$.

Combining the results from (\ref{dim1}) and (\ref{dim2}), we calculate
that the dimension of the subspace $\mathcal{P}_1(2m)$ is
\begin{equation}
  d_1(2m) + d_2(2m) 
  = \begin{cases}
    0\,, & \text{for $m = 0$} \\[0.3ex]
    2\binom{m + 2}{3} + \frac{m+1}{2}\,, & \text{for $m
               = 1,3,5,\dots$}\\[0.3ex]
    2\binom{m + 2}{3}\,, & \text{for $m = 2,4,6,\dots$}
    \end{cases}
\end{equation}
The vector subspace of degree $l$ harmonic polynomials which transform
according to (\ref{transprop-a}) is $\mathcal{H}_1(l) = \mathcal{P}_l(l) \cap
\mathcal{H}(l)$.  Using the fact that $\triangle(\mathbb{R}^4)\,{:}\;
\mathcal{P}_1(l) \longrightarrow \mathcal{P}(l - 2)$ is an onto linear map, we
obtain that the dimension of $\mathcal{H}_1(2m)$ is
\begin{equation}
   \dim\left(\mathcal{P}_1\left(2m\right)\right) 
      - \dim\left(\mathcal{P}_1\left(2m - 2\right)\right) 
  = \begin{cases}
    \frac{m(2m + 1)}{2}\,, & \text{for $m = 0,2,4,\dots$}\\
    \frac{(m + 1)(2m + 1)}{2}\,, & \text{for $m = 1,3,5,\dots$}
    \end{cases}
\label{dim-pi-1}
\end{equation}

Notice that each of the eigenspaces, $\mathcal{H}(l)$, is nontrivial,
except for $l = 0$.  This implies that the spectrum of
$-\triangle_1(Y_K)$ consists of all eigenvalues of
$-\triangle(S^3(r))$, except for the zero eigenvalue.
Re-indexing~(\ref{dim-pi-1}) in terms of $n = m-1$ completes the proof
for $-\triangle_1(Y_K)$.  The results for $-\triangle_2(Y_K)$ and
$-\triangle_3(Y_K)$ are obtained by simply permuting $\mathbf{i}$,
$\mathbf{j}$, and $\mathbf{k}$ in the above calculation.
\end{proof}

\section{Spectrum with Isoparametric Geometry}

To compute the spectra of the line bundle Laplacians on
$\mathfrak{F}(b)$ with the isoparametric metric, $g_b$, we shall make
use of the fact that the 3-dimensional unit sphere may be viewed as
$\Sp(1)$, the Lie group of unit quaternions in $\mathbb{H}$.  Note
that this Lie group is isomorphic to $\SU(2)$, should the reader
prefer to view it in that guise.

The Lie algebra, $\symp(1)$, of $\Sp(1)$ is a 3-dimensional real vector
space.  This Lie algebra may be viewed as the tangent space at the
identity and in this context it is a vector space generated by
$\left\{\mathbf{i}, \mathbf{j}, \mathbf{k}\right\}$.  For $\mathbf{x}
= x_1\mathbf{i} + x_2\mathbf{j} + x_3 \mathbf{k}$ and $\mathbf{y} =
y_1\mathbf{i} + y_2\mathbf{j} + y_3 \mathbf{k}$, we define the inner
product
\begin{equation}
  \mathbf{x} \cdot \mathbf{y} = \frac{1}{2}\sum_{i=0}^3 x_i y_i \,.
\end{equation}
With this inner product, the adjoing representation of $\Sp(1)$ on
 $\symp(1)$ defines a map
\begin{equation}
  \mathrm{Ad}:\, \Sp(1) \longrightarrow \SOrth(3) \subset
  \GL(\symp(1))\,,
  \label{Adrep}
\end{equation}
which is a 2-fold covering.  In the 3-fold coverings
of~(\ref{3cover}), $P(X_i) \cong \frac{\Orth(3)}{\Orth(1) \times
\Orth(1)}$ and $\mathfrak{F}(b) \cong \frac{\Orth(3)}{\Orth(1) \times
\Orth(1) \times \Orth(1)}$ are homogeneous spaces.  The projections
\begin{equation}
  P(\wp_i):\, P(X_i) \longrightarrow \mathfrak{F}(b)
\end{equation}
correspond to factoring out one of the $\Orth(1)$ factors, for each of
$i=1,2,3$, respectively.  The remaining 4-fold covering projections is
\begin{equation}
  c_i = \mathrm{Ad} \circ a_i\,,\quad \text{for $i=1,2,3$.}
\end{equation}
The 2-fold covering projection $a_i$ is defined as
\begin{equation}
  a_i : \, \SOrth(3) \longrightarrow P(X_i) \cong 
           \frac{\Orth(3)}{\Orth(1) \times \Orth(1)} =
           \frac{\SOrth(3)} {\Simple(\Orth(1) \times \Orth(1))}\,,
\end{equation}
where $\Simple(\Orth(1) \times \Orth(1)) =
\left\{(1,1)\,,\,(-1,-1)\right\}$. 

We shall define the Riemannian metric on $\Sp(1)$ to be $p^*g_b$, the
pull up of the Riemannian metric from $\mathfrak{F}(b)$.  Recall that
$g_b$ is invariant with respect to the adjoint action of $\Orth(3)$ on
$\mathfrak{F}(b)$.  Therefore, $p^*g_b$ is also an invariant metric and
the Laplacian with respect to $p^*g_b$ is an invariant differential
operator.  An invariant differential operator on a compact Lie group
is completely determined by its form on the Lie algebra~\cite{T94}.
Specifically, if $\left\{\mathbf{x}_1, \mathbf{x}_2, \mathbf{x}_3
\right\}$ is an orthonormal basis for $\left(\symp(1),p^*g_b\right)$,
then
\begin{equation}
  \triangle(\symp(1),p^*g_b) = \sum_{i=1}^3 \mathbf{x}_i^2\,, 
\end{equation}
where the tangent vector $\mathbf{x}_i$ acts on functions as a
directional derivative.

The differential of the projection $p$ is
\begin{equation}
  dp : \, \symp(1) \longrightarrow
           T_{D(b,r)}\mathfrak{F}(b)\,, \qquad
          \mathbf{x} \longmapsto
          \left[\mathrm{ad}(\mathbf{x}),D(b,r)\right]\,,
\end{equation}
where $\mathrm{ad}$, the adjoint representation of $\symp(1)$, is the
differential of $\mathrm{Ad}$ in equation~(\ref{Adrep}).  It follows
that the metric $p^*g_b$ on $\symp(1)$ is
\begin{equation}
  p^*g_b(\mathbf{x},\mathbf{y}) = 
  \trace\left([\mathrm{ad}(\mathbf{x}),D(b,r)]\,
              [\mathrm{ad}(\mathbf{y}),D(b,r)]\right)\,,
\end{equation}
where $D(b,r)$ is defined in equation~(\ref{diag-matrix}).
Therefore, an orthonormal basis for\linebreak[4] $(\symp(1),p^*g_b)$ is
\begin{equation}
  \mathbf{x}_1 = \frac{\sqrt{1-b+b^2}}{\sqrt{3}\,r(1-b)}\,\,
                 \mathbf{i}\,,\quad
  \mathbf{x}_2 = \frac{\sqrt{1-b+b^2}}{\sqrt{3}\,r}\,\,\mathbf{j}\,,
                 \;\text{and}\;\;\;
  \mathbf{x}_3 = \frac{\sqrt{1-b+b^2}}{\sqrt{3}\,rb}\,\,
                 \mathbf{k}
\end{equation}
In terms of this basis, the Laplacian is
\begin{equation}
  \triangle(\Sp(1),p^*g_b) 
  = \frac{1-b+b^2}{3r^2}\left(\left(1-b\right)^{-2}\mathbf{i}^2 +
    \mathbf{j}^2 + b^{-2}\mathbf{k}^2\right)
  \label{sp-lap}
\end{equation}

Our approach to determining the spectrum of $\triangle(\Sp(1),p^*g_b)$
will be to utilize the Peter-Weyl theorem as in~\cite{T94}.  The set
of all equivalence classes of irreducible representations of $\Sp(1)$
is
\begin{equation}
  \mathcal{R}(\Sp(1)) =
  \left\{(\rho_m,\mathcal{Q}_m)\mid m=0,1,2\dots\right\}\,,
\end{equation}
where $\mathcal{Q}_m$ is the vector space of homogeneous polynomials
in two complex variables \cite{S90}.  For $\mathbf{x} = x_0 +
x_1\mathbf{i} + x_2\mathbf{j} + x_3\mathbf{k} \in \mathbb{H}$, we
define two maps from $\mathbb{H}$ to $\mathbb{C}$ by
\mbox{$h_1(\mathbf{x}) = x_0 + i x_2$} and $h_2(\mathbf{x}) = x_1 + i
x_3$.  In terms of $h_1$ and $h_2$, a right action of $\mathbb{H}$ on
$\mathbb{C}^2$ is given by {\renewcommand{\arraystretch}{1.8}
\setlength{\arraycolsep}{0.5em}
\begin{equation}
  [z_1\,z_2] \longmapsto
    [z_1\,z_2]\,
      \begin{bmatrix} 
        \;\overline{h_1}(\mathbf{x})\; & \;-h_2(\mathbf{x})\;\\
        \;\overline{h_2}(\mathbf{x})\; & \;h_1(\mathbf{x})\;
      \end{bmatrix} 
\end{equation}
We then define the action of $\rho_m$ on $\mathcal{Q}_m$ by
\begin{equation}
  \rho_m(\mathbf{x}) f(z_1,z_2) =
  f(\,\overline{h_1}(\mathbf{x})\,z_1 + \overline{h_2}(\mathbf{x})\,z_2\,,\,
  h_1(\mathbf{x})\,z_2 - h_2(\mathbf{x})\,z_1)\,.
\end{equation}}

On $\mathcal{Q}_m$, we define the $\Sp(1)$-invariant
inner product
\begin{multline}
  ((f_1\,,\,f_2)) = \sum_{k=0}^m\overline{a_k}\,b_k\, k!\, (m-k)!\,,\\
  \text{where}\quad f_1=\sum_{k=0}^ma_kz_1^kz_2^{m-k}\quad
  \text{and}\quad f_1=\sum_{k=0}^mb_kz_1^kz_2^{m-k}\,.
\end{multline}
With respect to this inner product, an orthonormal basis for
$\mathcal{Q}_m$ is
\begin{equation}
  \left\{u^m_k = [\,k!\,(m-k)!\,]^{-\frac{1}{2}}\,z_1^k\,z_2^{m-k}
                  \mid k = 0,1,2,\dots m \right\}\,.
\end{equation}
The matrix elements of $\rho_m(\mathbf{x})$ are 
\begin{equation}
  [\rho_m(\mathbf{x})]_{ij} = (( u^m_i \,,\, \rho_m(\mathbf{x})\,
                                 u^m_j ))\,.
\end{equation}
It follows directly from the Peter-Weyl theorem that
\begin{equation}
  \left\{\sqrt{m+1}\>[\rho_m]_{ij} \mid i,j = 0,1,2,\dots m 
          \quad\text{and}\quad m = 0,1,2,\dots\right\}
\end{equation}
is an orthonormal basis for $L^2(\Sp(1); \mathbb{C})$.  The Hilbert
space inner product on\linebreak[4] $L^2(\Sp(1);\mathbb{C})$ is
\begin{equation}
  \langle f \,,\, g \rangle = \int \overline{f(\mathbf{x})}\,
                              g(\mathbf{x})\,d\mathbf{x}\,,
\end{equation}
where $d\mathbf{x}$ is the usual normalised bi-invariant Haar measure
on $\Sp(1)$.  Note that this measure is proportional to the measure
induced by the Riemannian metric $p^*g_b$.

The representation $(\rho_m,\mathcal{Q}_m)$ of $\Sp(1)$
induces an infinitesimal  representation
$(\rho^\prime_m,\mathcal{Q}_m)$ of the Lie algebra
$\symp(1)$.  In terms of the basis polynomials $u_{mk}$ for
$\mathcal{Q}_m$, this representation is given by 
\begin{subequations}
\begin{align}
  \rho^\prime_m(\mathbf{i})\,u^m_k & = k u^m_{k-1} - (m - k)u^m_{k+1}\\
  \rho^\prime_m(\mathbf{j})\,u^m_k & = i(m-2k)u_k^m \\
  \rho^\prime_m(\mathbf{k})\,u^m_k & = -i\left(ku^m_{k-1} + 
                                       (m - k)u^m_{k+1}\right)
\end{align}
\end{subequations}

Utilizing~(\ref{sp-lap}), the Casimir operator for
$-\triangle(\Sp(1),p^*g_b)$ is
\begin{multline}
  \rho^\prime_m(-\triangle(\Sp(1),p^*g_b))\\
       = \frac{1-b+b^2}{3r^2}\left(\left(1-b\right)^{-2}
         \rho^\prime_m(\mathbf{i})^2 + 
       \rho^\prime_m(\mathbf{j})^2 + 
       b^{-2}\rho^\prime_m(\mathbf{k})^2\right),
\end{multline}
for the representation $(\rho_m,\mathcal{Q}_m)$.  Writing
this as an $(m+1) \times (m+1)$ matrix with respect to the basis
$\{u_i\}$, we define
\begin{equation}
  \left[\Omega_m\right]_{ij} = (( u^m_i
  \,,\,\rho^\prime_m(-\triangle(\Sp(1),p^*g_b))\,
  u^m_j))\,,\quad\text{for $i,j = 0,1,2,\dots,m$.}
\end{equation}
It follows from the Peter-Weyl theorem that the eigenvalues of the
Laplacian\linebreak[4] $-\triangle(\Sp(1),p^*g_b)$ are given by all of the
eigenvalues of the matrices $\Omega_m$, for\linebreak[4]
\mbox{$m=0,1,2,3,\dots$}.
Furthermore, if $\lambda^m$ is an eigenvalue of $\Omega_m$, then the
multiplicity of $\lambda^m$ as an eigenvalue of
$-\triangle(\Sp(1),p^*g_b)$ is $\dim(\mathcal{Q}_m) = m+1$.

Note that the restriction of the action $\rho_m$ to the unit
quaternions $\quat \subset \mathbb{H}$ corresponds to the action
$\Lambda$ defined in $\ref{transa}$, where we identify functions of
two complex variables with functions of four real variables by
\begin{equation}
  f(z_1,z_2) = f\bigl(\Re(z_1),\Re(z_2),\Im(z_1),\Im(z_2)\bigr)\,.
\end{equation}
The real and imaginary parts of a complex number are denoted by $\Re$
are $\Im$, respectively.

The action $\Lambda$ of $\quat$ on
$\mathcal{Q}_m$ is simply the restriction of $\rho_m$ to
$\quat \subset \mathbb{H}$.  Therefore,
\begin{subequations}
\begin{align}
  \Lambda_\mathbf{\pm 1} & :\;  f(z_1,z_2)
     \longmapsto f(\pm z_1, \pm z_2) \\
  \Lambda_\mathbf{\pm i} & :\;  f(z_1,z_2)
     \longmapsto f(\pm z_2, \mp z_1) \\
  \Lambda_\mathbf{\pm j} & :\;  f(z_1,z_2)
     \longmapsto f(\mp i z_1, \pm i z_2) \\
  \Lambda_\mathbf{\pm k} & :\;  f(z_1,z_2)
     \longmapsto f(\mp z_2 i , \mp z_1 i)
\end{align}
\end{subequations}
This action commutes with the endomorphism
$\rho^\prime_m(-\triangle(\Sp(1),p^*g_b))$, because it is an isometry
with respect to the inner product $((\,\cdot\, , \,\cdot\,))$.
Therefore, by considering the subspace of $\mathcal{Q}_m$ that
transforms according to \ref{transprop-a} under the action $\Lambda$,
we obtain the Casimir matrices $\Omega^1_m$ of
$-\triangle_1\left(\mathfrak{F}(b)\right)$. The result is summarised
in the following theorem.

\begin{theorem}
\label{casimir1}
For $m = 1,3,5,\dots$, the Casimir matrix $\Omega_m^1$ is an
$\frac{m+1}{2} \times \frac{m+1}{2}$ matrix. It's entries, which are labelled
by $i,j = 0,1,2,\dots,\frac{m-1}{2}$, are
\begin{equation}
  \left[\Omega^1_m\right]_{ij} = 
  \begin{cases}
    \dfrac{1 - b + b^2}{3}\left[\dfrac{(m - 1)(m + 2)}{4\,b^2} +
      \dfrac{3m^2 + 3m - 2}{4\,(1-b)^2} + 1\right], & \\
      & \hspace*{-7em}\mbox{for $i = j = \dfrac{m - 1}{2}$} \\[1ex]
    \dfrac{1 - b + b^2}{3}\left[\dfrac{1 - 2b + 2b^2}{2\,b^2\,(1 - b)^2}
      \bigl((4i)(m-i) + m\bigr) + (2i - m)^2\right], & \\
      & \hspace*{-7em}\mbox{for $i = j < \dfrac{m - 1}{2}$} \\[1ex]
    \dfrac{(1 - b + b^2)(1 - 2b)}{6\,b^2\,(1-b)^2}\sqrt{(i + 1)(m - i)(2i
      + 1)(2m - 2i - 1)}\,, & \\ 
      & \hspace{-7em} \mbox{for $i = j - 1$} \\[1ex]
    \dfrac{(1 - b + b^2)(1 - 2b)}{6\,b^2\,(1-b)^2}\sqrt{(i)(m - i + 1) (2i
      - 1) (2m - 2i + 1)}\,, & \\ 
      & \hspace*{-7em}\mbox{for $i = j+1$}
      \\
    0\,, & \hspace*{-7em}\text{otherwise}
  \end{cases}
\smallskip
\end{equation}
For $m = 2,4,6,\dots$, the Casimir matrix $\Omega_m^1$ is an
$\frac{m}{2} \times \frac{m}{2}$ matrix. For $i,j =
0,1,\dots,\frac{m-2}{2}$, it entries are
\begin{equation}
  \left[\Omega^1_m\right]_{ij} = 
  \begin{cases}
    \dfrac{1 - b + b^2}{3}\left[\dfrac{3m^2 + 3m - 2}{4\,b^2} +
      \dfrac{(m - 1)(m + 2)}{4\,(1-b)^2} + 1\right], &  \\
         & \hspace*{-9em} \text{for $i = j = \dfrac{m - 2}{2}$} \\[1ex]
      \dfrac{1 - b + b^2}{3}\biggl[\dfrac{1 - 2b + 2b^2}{2\,b^2\,(1 - b)^2}
      \bigl((2i + 1)(2m - 2i + 1) + m\bigr) \\[1ex]
          \hspace*{18em}+(2i +1 - m)^2\biggr], \\[2ex]
	  & \hspace*{-9em}\text{for $i = j < \dfrac{m - 2}{2}$} \\
    \dfrac{(1 - b + b^2)(1 - 2b)}{6\,b^2\,(1-b)^2}\sqrt{(i + 1)(m - i - 1)
      (2i + 3)(2m - 2i - 1)}\,, & \\
         & \hspace*{-9em}\text{for $i = j - 1$}\\[1ex]
    \dfrac{(1 - b + b^2)(1 - 2b)}{6\,b^2\,(1-b)^2}\sqrt{(i)(m - i) (2i
      + 1) (2m - 2i + 1)}\,,& \\[1ex]
         & \hspace*{-9em}\text{for $i = j+1$} \\
    0\,, & \hspace*{-9em}\text{otherwise}
  \end{cases}
\end{equation}
\end{theorem}

Note that matrices $\Omega^1_m$ are tridiagonal in that only the diagonal,
super-diagonal, and sub-diagonal entries are non-zero.  This makes a
numerical computation of the eigenvalues quite computationally
inexpensive.  Furthermore, At $b = \frac{1}{2}$ the matrices $\Omega^1_m$
are diagonal and explicit formulae for the eigenvalues of
$-\triangle_1\left(\mathfrak{F}\left(\frac{1}{2}\right)\right)$ are
easily obtained.  After some simplification, the eigenvalues are
\begin{equation}
  \lambda_{m,l} = \frac{1}{4}\,m^2 + 3\,l\,m - 3\,l^2 + m\,,
\end{equation}
for $m = 1,3,5,\dots$ and $l = 0,1,2,\dots,\frac{m-1}{2}$.  Also,
\begin{equation}
  \lambda_{m,l} = \frac{1}{4}\,m^2 + 3\,l\,m - 3\,l^2 + \frac{5}{2}\,m
  -3\,l - \frac{3}{4}
\end{equation}
for $m = 2,4,6,\dots$ and $l = 0,1,2,\dots,\frac{m-2}{2}$.

The Casimir matrices for the Laplacian $-\triangle_3(\mathfrak{F}(b))$
on the line bundle $\eta_3|_{\mathfrak{F}(b)}$ are obtained from the
Casimir matrices for $-\triangle_1(\mathfrak{F}(b))$ by substituting
$1-b$ for $b$.

The Casimir matrices of $-\triangle_2(\mathfrak{F}(b))$ are obtained
from the subspace of $\mathcal{Q}_m$ that transforms as
\ref{transprop-b} under the action $\Lambda$. The result is given in
the following theorem.

\begin{theorem}
\label{casimir2}
For $m = 1,2,3,\dots$, the $ij$ entry of the $\frac{m+1}{2} \times
\frac{m+1}{2}$ matrix $\Omega^2_m$ is
\begin{equation}
\left[\Omega^2_m\right]_{ij} = 
  \begin{cases}
    \dfrac{1-b+b^2}{3}\biggl[\bigl(\left(4i\right)\left(m - i - 1\right)
        + 3m - 1\bigr)\left(\dfrac{1 - 2\,b +
        2\,b^2}{2\,b^2\,(1-b)^2}\right) \\
     \hspace*{18em}+ \left(m - 2i - 1\right)^2\biggr]\,, \\[2ex]]
    & \hspace*{-9em}\text{for $i = j = \dfrac{m - 1}{2}$}\\[1ex]
  \dfrac{(1- b + b^2)(1 - 2b)}{6\,b^2\,(1-b)^2}\sqrt{(i+1)(2i+3)
    (m - i - 1)(2m - 2i - 1)}\,, & \\
    & \hspace*{-9em}\text{for $j = i + 1 \le \dfrac{m - 3}{2}$} \\[1ex]
  \dfrac{(1- b + b^2)(1 - 2b)}{6\,b^2\,(1-b)^2}\sqrt{(i)(2i + 1)
      (m - i)(2m - 2i + 1)}\,, & \\
    & \hspace*{-9em}\text{for $j = i - 1 \le \dfrac{m - 5}{2}$} \\[1ex]
  \dfrac{(1- b + b^2)(1 - 2b)}{6\sqrt{2}\,b^2\,(1-b)^2}\sqrt{(m-1) (m)
    (m+1)(m+2)}\,, & \\
    & \hspace*{-17em}\text{for $j = i = \dfrac{m-1}{2}$ or $j = i - 1
    = \dfrac{m - 3}{2}$} \\[1ex]
  0\,, & \hspace*{-9em}\text{otherwise}
  \end{cases}
\end{equation}
For $m = 2,4,6,\dots$, the $ij$ entry of the $\frac{m}{2} \times
\frac{m}{2}$ matrix $\Omega_m^2$  is
\begin{equation}
\left[\Omega^2_m\right]_{ij} = 
  \begin{cases}
    \dfrac{1-b+b^2}{3}\left[\bigl(\left(4i\right)\left(m - i\right) +
      m\bigr)\dfrac{1 - 2b + b^2}{2\,b^2\,(1-b)^2} + \left(m -
      2i\right)^2 \right], & \\
      & \hspace*{-11em} \text{for $i = j = 0,1,2,\dots,\dfrac{m - 2}{2}$} \\[1ex]
    \dfrac{(1-b+b^2)(1-2b)}{6\,b^2(1-b)^2}\sqrt{(i+1)(2i+1)(m-i)(2m-2i-1)}\,,
      & \\
    & \hspace*{-11em}\text{for $j = i+1$} \\[1ex]
    \dfrac{(1-b+b^2)(1-2b)}{6\,b^2(1-b)^2}\sqrt{(i)(2i - 1)(m-i+1)
      (2m-2i+1)}\,, & \\
      &\hspace*{-11em} \text{for $j = i-1$} \\[1ex]
      0\,, & \hspace*{-11em}\text{otherwise}
  \end{cases}
\end{equation}
\end{theorem}

At $b = \frac{1}{2}$ the eigenvalues of
$-\triangle_2\left(\mathfrak{F}\left(b\right)\right)$ can be given
explicitly because the matrices $\Omega_m^2\left(\frac{1}{2}\right)$ are
diagonal.  For $m = 1,3,5,\dots$ the eigenvalues are
\[
  \lambda^2_{m,l} = \frac{1}{4}m^2 + 3\,l\,m - 3\,l^2
  + \frac{5}{2}\,m - 3l - \frac{3}{4}\,,
\]
where $l = 0,1,2\dots,\frac{m-1}{2}$.  Also, for $m = 2,4,6,\dots$ the
eigenvalues are
\[
  \lambda^2_{m,l} = \frac{1}{4}\,m^2 + 3\,l\,m - 3\,l^2 + m\,,
\]
where $l = 0,1,2,\dots,\frac{m-2}{2}$.

\section{Spectral Flow}

The spectra of the tridiagonal Casimir matrices $\Omega^1_m$ and
$\Omega^2_m$ have been numerically computed as a function of $b$ and
the results plotted in Figures~\ref{spectrum1} and~\ref{spectrum2},
respectively. These plots show the spectral flow of
$-\triangle_1(\mathfrak{F}(b))$ and $-\triangle_2(\mathfrak{F}(b))$
with $b$. The eigenvalues of $-\triangle_i(\mathfrak{F}(b))$ are
denoted by $\lambda^i_{m,l}$, which is the $l^{\rm th}$ eigenvalue of
$\Omega^i_m$.

First consider the spectral flow of $-\triangle_1(\mathfrak{F}(b))$
shown in Figure~\ref{spectrum1}.  As $b$ goes to $0$, the singular
foliation $\mathfrak{F}(b)$ flows to the sphere
$\mathfrak{R}_1$. Under this flow, most of the eigenvalues of
$-\triangle(\mathfrak{F}(b))$ blow-up to infinity. However, the
eigenvalues $\lambda^1_{m,0}$ for $m=1,3,5,\dots$ flow to the
eigenvalues of $-\triangle_1(\mathfrak{R_1})$ calculated in
Theorem~\ref{thmR1}.  As $b$ goes to $1$, all eigenvalues of
$-\triangle_1(\mathfrak{F}(b))$ blow up. This is to be expected,
because the line bundle $\eta_1$ is not defined on $\mathfrak{R_3}$.

Now consider the spectral flow of $-\triangle_2(\mathfrak{F}(b))$, shown
in Figure~\ref{spectrum2}. The line bundle $\eta_2$ is not defined in
either of the $b \to 0$ or $b \to 1$ limits. This results in the
all eigenvalues of $-\triangle_2(\mathfrak{F}(b))$ blowing up to
infinity as $b$ goes to either $0$ or $1$.

\begin{figure}
\begin{center}
\includegraphics{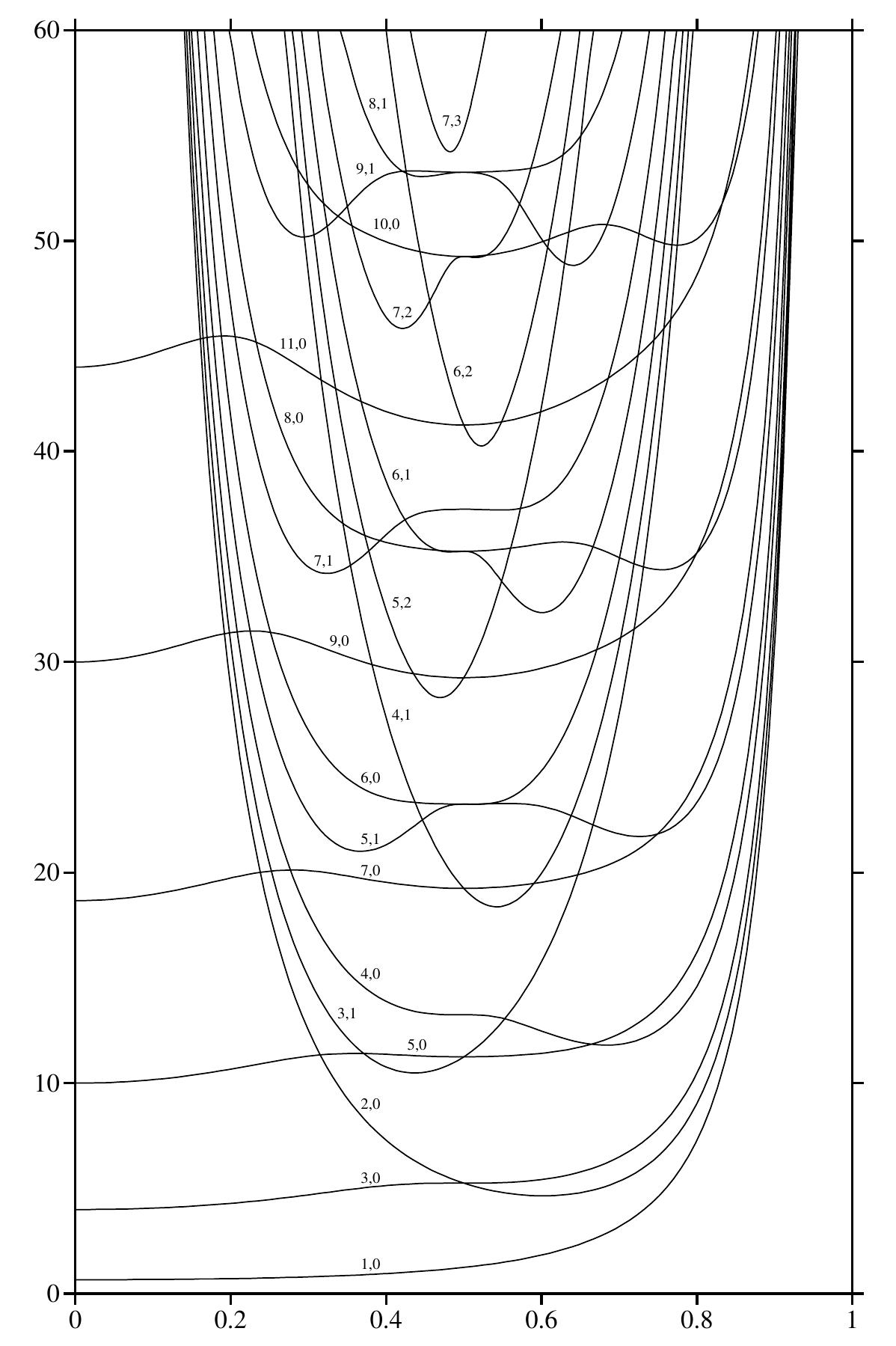}
\end{center}
\caption{The eigenvalues of $-\triangle_1(\mathfrak{F}(b))$ plotted
against $b$.  The eigenvalue $\lambda^1_{m,l}$ is labelled by $m,j$ on
the graph.}
\label{spectrum1}
\end{figure}

\begin{figure}
\begin{center}
\includegraphics{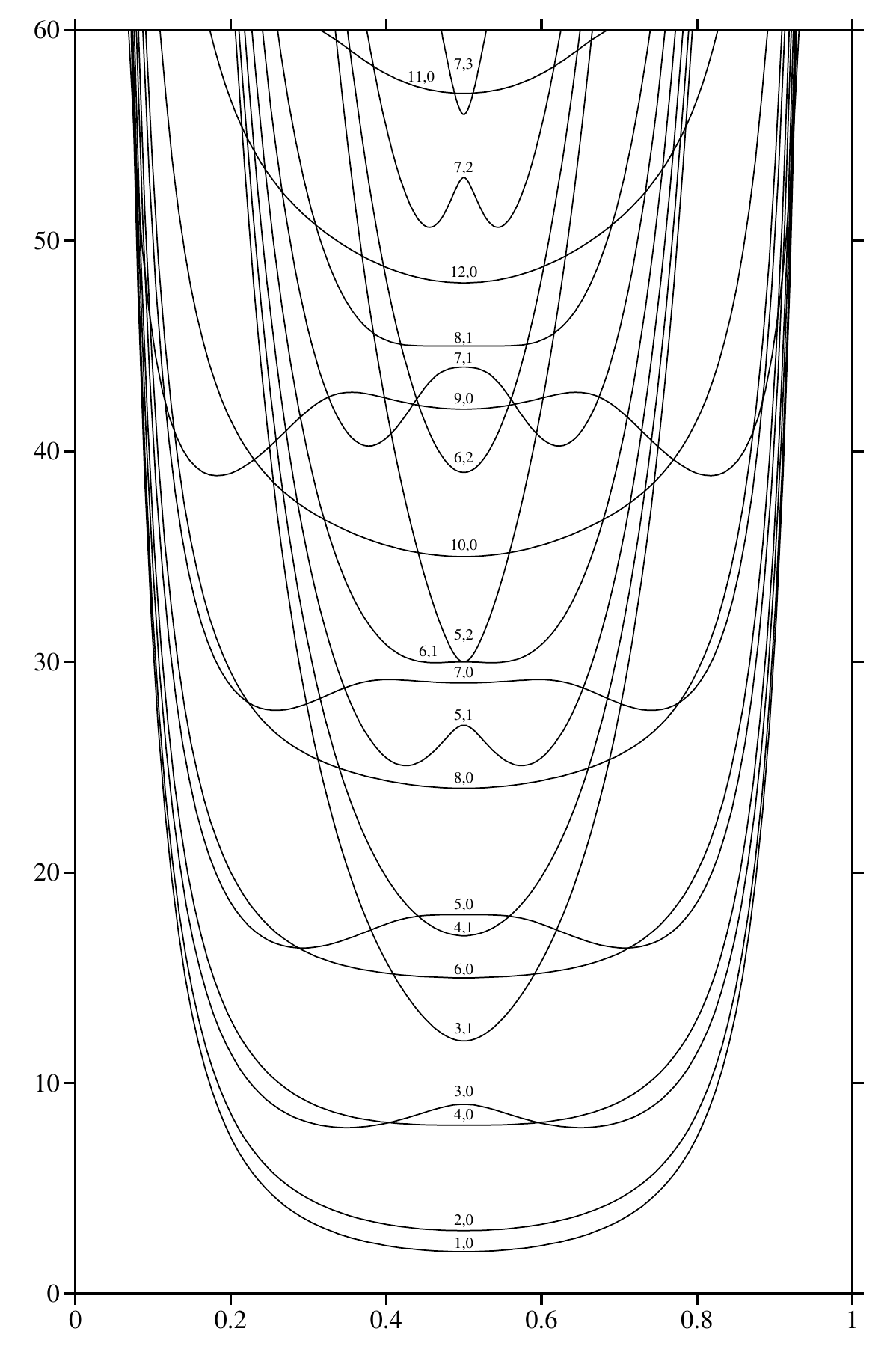}
\end{center}
\caption{The eigenvalues of $-\triangle_2(\mathfrak{F}(b))$ plotted
against $b$.  The eigenvalue $\lambda^2_{m,l}$ is labelled by $l,m$ on
the graph.}
\label{spectrum2}
\end{figure}

\section{Conclusions}

An electronic triplet, $\{\lambda_1,\lambda_2,\lambda_3\}$ has three
Born-Openheimer Hamiltonians, $B_i$, associated with each of the three
eigenvalues on the triplet. The potential function for each of these
Hamiltonians is minimised on one of isoparametric submanifolds
$\mathfrak{F}(b)$.  Therefore, up to a constant, the spectrum of $B_i$
is given by the spectrum of $-\triangle_i(\mathfrak{F}(b))$, for $b =
\frac{\lambda_2 - \lambda_1}{\lambda_3 - \lambda_1}$.  These
eigenvalues are given by the eigenvalues of the Casimir matrices in
Theorems~\ref{casimir1} and~\ref{casimir2}. They are plotted in
Figures~\ref{spectrum1} and~\ref{spectrum2}.

Although we have specifically considered the $T \otimes (e_g \oplus
t_{2g})$ Jahn-Teller effect of octahedral molecules, our results also
apply to pseudorotational excitations of electronic triplet states for
any molecule.  This includes the $T \otimes h_g$ Jahn-Teller effect of
icosahedral molecules.  Evidence for this Jahn-Teller effect in the
icosahedral molecule ${\rm C}_{60}$ has been reported in~\cite{GAS92}.
As well, molecules with icosahedral symmetry also exhibit a
Jahn-Teller effect for electronic quadruplets.

We remark on the relationship between our work and Berry
phases in Jahn-Teller problems.  Berry phases appearing in the $T
\otimes \left( e_g \oplus t_{2g}\right)$ and $T \otimes h_g$
Jahn-Teller problems have been studied in \cite{CO88}
and~\cite{AMT94,MTA94}, respectively.  The topological results in this
paper can also be used to compute Berry phases and we refer
to~\cite{DR88} for a description of how this is done.  However,
it is apparent from the results in this paper that the full
geometrical content of the Jahn-Teller effect exceeds that which can
be determined by Berry phases alone.

To summarise, the primary contribution of our work has been to supply
a rigorous basis for the topological and geometrical aspects of
Jahn-Teller computations for electronic triplets.  We initially
reported on the connection between isoparametric geometry and the
Jahn-Teller effect in \cite{R07}.  In pursuing this relationship, we
have demonstrated that the connection Laplacian on the tautological
lines bundles over Cartan's isoparametric foliation of type~3 in $S^3$
exhibit interesting spectral flow properties. 

\FloatBarrier

\subsection*{Acknowledgment}

I owe a great debt to Roy R.~Douglas for his mentorship and
collabration during the early stages of this work.  I thank
M.~C.~M.~O'Brien and J.~W.~Zwanziger for comments about experimental
applications of our work.  I thank R.~Mazzeo for recommending the
interesting work of B.~Solomon on harmonic analysis of isoparametric
minimal hypersurfaces.

\end{document}